\documentclass[conference]{IEEEtran}
\IEEEoverridecommandlockouts

\usepackage{cite}
\usepackage{amsmath,amssymb,amsfonts}
\usepackage{algorithmic}
\usepackage{graphicx}
\usepackage{textcomp}
\usepackage{xcolor}
\usepackage{amsmath, bm}
\usepackage{amsfonts}
\usepackage{amssymb}
\usepackage{balance}

\def\BibTeX{{\rm B\kern-.05em{\sc i\kern-.025em b}\kern-.08em
    T\kern-.1667em\lower.7ex\hbox{E}\kern-.125emX}}

\newcommand{\comment}[1]{}
\newtheorem{theorem}{Theorem}
\newtheorem{lemma}{Lemma}
\newtheorem{proposition}{Proposition}

\newtheorem{definition}{Definition}
\newtheorem{remark}{Remark}

\def\bp{{R}} 
\def\br{{R}} 
\def\mc{{\mathsf{T}}}  
\def\tp{{P}} 

\begin{document}

\title{Second-Order Asymptotics of Two-Sample Tests\\
\thanks{T.~Koch has received funding from the Spanish Ministerio de Ciencia, Innovación y Universidades under Grant PID2024-159557OB-C21 (MICIU/AEI/10.13039/501100011033 and ERDF/UE) and from the Comunidad de Madrid under Grant IDEA-CM (TEC-2024/COM-89). J.~Ravi has  received funding from IIT Kharagpur under a Faculty Start-up Research Grant (FSRG).}
}


\author{%
  \IEEEauthorblockN{K V Harsha\IEEEauthorrefmark{1}, Jithin Ravi\IEEEauthorrefmark{2}, and Tobias Koch\IEEEauthorrefmark{4}}
  \IEEEauthorblockA{\IEEEauthorrefmark{1}%
  Birla Institute of Technology and Science, Pilani, Hyderabad, India}
  \IEEEauthorblockA{\IEEEauthorrefmark{2}%
  Indian Institute of Technology Kharagpur, Kharagpur, India}
     \IEEEauthorblockA{\IEEEauthorrefmark{4}%
             Universidad Carlos III de Madrid, Legan\'es, Spain, and Gregorio Mara\~n\'on Health Research Institute, Madrid, Spain}
             \IEEEauthorblockA{Emails: harsha.kv@hyderabad.bits-pilani.ac.in, jithin@ece.iitkgp.ac.in, tkoch@ing.uc3m.es}
}

\maketitle

\begin{abstract}
In two-sampling testing, one observes two independent sequences of independent and identically distributed random variables distributed according to the distributions $P_1$ and $P_2$ and wishes to decide whether $P_1=P_2$ (null hypothesis) or $P_1\neq P_2$ (alternative hypothesis). The Gutman test for this problem compares the empirical distributions of the observed sequences and decides on the null hypothesis if the Jensen-Shannon (JS) divergence  between these empirical distributions is below a given threshold. This paper proposes a generalization of the Gutman test, termed \emph{divergence test}, which replaces the JS divergence by an arbitrary divergence. For this test, the exponential decay of the type-II error probability for a fixed type-I error probability is studied. First, it is shown that the divergence test achieves the optimal first-order exponent, irrespective of the choice of divergence. Second, it is demonstrated that divergence tests with invariant divergences achieve the same second-order asymptotics as the Gutman test. In addition, a connection between two-sample testing and robust goodness-of-fit testing is established. 
\end{abstract}


\section{Introduction}
Two-sample testing is a fundamental problem in statistical hypothesis testing. In this problem, we observe two independent and identically distributed (i.i.d.) sequences $X^n$ and $Y^n$, generated according to the distributions $P_1$ and $P_2$, respectively. We wish to determine whether the two sequences are sampled from the same underlying distribution. This leads to the following binary hypothesis testing problem:
\begin{align*}
	H_0: P_1=P_2 \quad \text{vs.} \quad H_1: P_1 \neq P_2.
\end{align*}
The two-sample testing problem has a long history in statistics and machine learning. A classical test for this problem was proposed by Kolmogorov and Smirnov in the 1930s, known as the Kolmogorov–Smirnov (KS) test \cite{LR205}. The KS test compares the empirical cumulative distributions of the two sequences and rejects the null hypothesis when their maximum difference exceeds a threshold. A closely related problem was considered by Gutman \cite{Gut89}, where two training sequences corresponding to distributions $P_1$ and $P_2$ are given, and the goal is to determine whether an additional test sequence was generated according to $P_1$ or $P_2$. This problem is particularly relevant in supervised learning, since many binary classification problems can be formulated as the instances of the Gutman test. In the machine learning literature, two-sample testing has also been studied via kernel-based methods, most notably through the Maximum Mean Discrepancy (MMD), where the MMD between the two observed sequences is compared against a threshold \cite{Gretton12}.

The asymptotic performance in Stein's regime of MMD-based two-sample testing was analyzed in \cite{Shengyu21}. Specifically, the authors studied the decay rate of the type-II error probability
$\beta_n$ as $n\to\infty$ given that the type-I error  $\alpha_n$ remains below a fixed $\epsilon\in(0,1)$. It was shown that the error exponent of $\beta_n$ of the MMD-based test is given by twice the Bhattacharyya distance. It was further shown that this exponent is, in fact, the largest achievable error exponent among all two-sample tests. A similar asymptotic analysis was carried out by Zhou, Tan and Motani for the Gutman test \cite{Lin19}, who showed that the Gutman test also achieves an error exponent equal to twice the Bhattacharyya distance. It is noted that the Gutman test can be interpreted as comparing the Jensen–Shannon (JS) divergence \cite[Eq.~(7.8)]{PW25} between the empirical distributions of $X^n$ and $Y^n$ to a threshold. They further studied the second-order asymptotics of the two-sample test under the assumption that the JS divergence between $P_1$ and $P_2$ exceeds a given $\lambda$ \cite{Lin19}. In particular, they characterized the first-order and second-order behaviors of the type-I error $\alpha_n$ given that the type-II error $\beta_n$ remains below a fixed $\epsilon\in(0,1)$.

It can be shown that, for distributions over finite alphabets, the MMD is a divergence in the sense defined in the information geometry literature \cite{AB216}. Thus, a natural question that arises is what happens when the MMD or the JS divergence in the two-sample test is replaced by another divergence measure. The main focus of this paper is to analyze the asymptotic performance in Stein’s regime of divergence-based two-sample tests, where the empirical distributions of $X^n$ and $Y^n$ are compared using a general divergence. In particular, we focus on the asymptotic behavior of the type-II error given that the type-I error  $\alpha_n$ is below a fixed $\epsilon\in(0,1)$.

We show that the divergence test achieves the optimal first-order error exponent irrespective of the divergence. We further establish the second-order asymptotics of divergence tests with \emph{invariant divergences} (see Definition~\ref{invariancediv}). The class of invariant divergences is large and includes the Rényi divergence and the $f$-divergences, which in turn include the JS divergence and the Kullback-Leibler (KL) divergence \cite[Ch.~7]{PW25}, \cite{CS204}. Our results demonstrate that divergence tests with invariant divergences achieve the same second-order asymptotics as the Gutman test.
In addition, we establish a connection between the Gutman test and the robust goodness-of-fit testing problem.

The rest of the paper is organized as follows. Section~\ref{sec:setting} introduces the problem setting. Section~\ref{sec:div_test} introduces the notion of divergence and presents the divergence test. Section~\ref{sec:GLRT} establishes a connection between two-sample testing and the robust goodness-of-fit testing problem. Section~\ref{sec:main_results} discusses the second-order asymptotics of divergence tests with invariant divergences. Finally, Section~\ref{sec:non-invariant} discusses the asymptotics of divergence tests with non-invariant divergences.

\section{Problem Formulation}
\label{sec:setting}

Consider two random variables $X$ and $Y$  taking values in the same discrete alphabet  $\mathcal{Z}=\{ a_{1}, \ldots, a_{k}\}$, where $ k \geq 2$. The probability distributions of $X$ and $Y$ are denoted by $k$-length vectors $P_1=(P_{11},\ldots,P_{1k})^{T}$ and  $P_2=(P_{21},\ldots,P_{2k})^{T}$, respectively.
In two-sample testing, we observe two independent length-$n$ sequences $X^{n}=(X_{1}, \ldots, X_{n})$ and  $Y^{n}=(Y_{1}, \ldots, Y_{n})$, where $X^n$ are i.i.d. according to $P_1$ and $Y^n$ are i.i.d. according to $P_2$. Our goal is to determine whether the underlying distributions are the same. More precisely, we consider the binary hypothesis testing problem with hypotheses $H_{0}\colon P_1=P_2$ and $H_1\colon P_1 \neq P_2$ and wish to design a universal test $\mathsf{T}_n: \mathcal{Z}^n \times \mathcal{Z}^n \to \{H_0, H_1\}$ that decides which hypothesis is true without using any knowledge of $P_1$ or $P_2$. Thus, under $H_0$, we have that $P_1 = P_2 = P$, where $P$ is an arbitrary distribution on $\mathcal{Z}$; under $H_1$, $P_1$ and $P_2$ are two different but otherwise arbitrary distributions on~$\mathcal{Z}$.

The hypothesis test $\mathsf{T}_n$ partitions the space $\mathcal{Z}^n \times \mathcal{Z}^n$ into two disjoint regions. Indeed, let $\mathcal{A}(\mathsf{T}_n)$ denote the set of tuples $(x^n,y^n)\in\mathcal{Z}^n\times\mathcal{Z}^n$ that favors $H_0$. The  type-I error $\alpha_{n}$ and type-II error $\beta_{n}$  are then defined as
\begin{subequations}
\begin{IEEEeqnarray}{lCl}
	\alpha_{n}(\mathsf{T}_{n}) & \triangleq & \sum_{(x^n,y^n)\notin \mathcal{A}(\mathsf{T}_n)} P(x^n)P(y^n) \\
	\beta_{n}(\mathsf{T}_{n}) & \triangleq & \sum_{(x^n,y^n)\in \mathcal{A}(\mathsf{T}_n)} P_1(x^n)P_2(y^n). 
\end{IEEEeqnarray}
\end{subequations}
Our goal is to analyze the asymptotic behavior (as $n\to\infty$) of the type-II error $\beta_{n}(\mathsf{T}_{n})$ when the type-I error satisfies $\alpha_{n}(\mathsf{T}_{n}) \leq \epsilon$ for some $0<\epsilon<1$.  We define the first-order term $\beta'$ and the second-order term $\beta''$ of $\beta_n(\mathsf{T}_n)$ as follows:
\begin{subequations}
\begin{IEEEeqnarray}{lCl}
	\beta' &\triangleq & \lim\limits_{n \rightarrow \infty} -\frac{1}{n}  \ln \beta_{n}(\mathsf{T}_n) \\
	\beta'' &\triangleq &  \lim\limits_{n \rightarrow \infty} \frac{-\ln \beta_n(\mathsf{T}_n) -n \beta'} {\sqrt{n}}
\end{IEEEeqnarray}
\end{subequations}
if these limits exist. Our objective is to characterize $\beta'$  and $\beta ''$ for divergence-based two-sample tests or, in short, \emph{divergence tests}, which are defined in Section~\ref{sec:div_test}.

\subsection{Order Notation}
Let $f(x)$ and $g(x)$ be two real-valued functions. For \mbox{$a \in \mathbb{R} \cup \{\infty\}$}, we write  $f(x)=O(g(x))$ as $x \rightarrow a$ if  \mbox{$\varlimsup_{x \to a} \frac{|f(x)|}{|g(x)|} < \infty$}.  Similarly, we  write  $f(x)=o(g(x))$ as $x \rightarrow a$ if \mbox{$\lim_{x \to a} \frac{|f(x)|}{|g(x)|} = 0$}. Throughout the paper, $\varlimsup$ denotes the \emph{limit superior} and $\varliminf$ denotes the \emph{limit inferior}.
\comment{
Following is for the long version.
\subsection{Notation}
Let $f(x)$ and $g(x)$ be two real-valued functions. For \mbox{$a \in \mathbb{R} \cup \{\infty\}$}, we write  $f(x)=O(g(x))$ as $x \rightarrow a$ if  \mbox{$\varlimsup_{x \to a} \frac{|f(x)|}{|g(x)|} < \infty$}.  Similarly, we  write  $f(x)=o(g(x))$ as $x \rightarrow a$ if \mbox{$\lim_{x \to a} \frac{|f(x)|}{|g(x)|} = 0$}.  Finally, we write $f(x) = \omega(g(x))$ as $x \to a$	 if $g(x)=o(f(x))$. Here, $\varlimsup$ denotes the \emph{limit superior} and $\varliminf$ denotes the \emph{limit inferior}.
Convergence in distribution is denoted by $\xrightarrow{d}$. We write $\mathcal{N}(\bm{\mu}, \bm{\Sigma})$ to either denote a Gaussian random vector with mean $\bm{\mu}$ and covariance matrix $\bm{\Sigma}$ or to denote its distribution.
}

\section{Divergence-Based Two-Sample Tests}
\label{sec:div_test}
Divergence tests compare the empirical distributions of $X^n$ and $Y^n$ and decide in favor of $H_0$ if the distributions are sufficiently close, with closeness measured by an arbitrary divergence, defined next.

\subsection{Divergence}

Let $\bar{\mathcal{P}}(\mathcal{Z})$ and $\mathcal{P}(\mathcal{Z}) $ denote the set of probability distributions on $\mathcal{Z}$ and the set of strictly positive probability distributions, respectively. Further let $\mathcal{P}^{2}(\mathcal{Z}) \triangleq\mathcal{P}(\mathcal{Z}) \times  \mathcal{P}(\mathcal{Z})$ and $\bar{\mathcal{P}}^{2}(\mathcal{Z}) \triangleq \bar{\mathcal{P}}(\mathcal{Z}) \times  \bar{\mathcal{P}}(\mathcal{Z})$. 
Every probability distribution \mbox{$R \in \mathcal{P}(\mathcal{Z})$} can be written as a length-$k$ vector $\bp=(\bp_{1},\ldots,\bp_{k})^{\mathsf{T}}$, which in turn can also be represented by its first $(k-1)$ components, denoted by the vector $\mathbf{\bp}=(\bp_{1}, \ldots, \bp_{k-1})^{\mathsf{T}}\in \Xi$, where $\Xi$ is the coordinate space
\begin{equation}
\Xi \triangleq \left\lbrace (R_{1},\ldots, R_{k-1})^{\mathsf{T}} \colon  R_{i} >0, \sum_{i=1}^{k-1} R_{i} <1 \right\rbrace.
\end{equation}
Given any two probability distributions $S, R \in \mathcal{P}(\mathcal{Z})$, 
a \emph{divergence} is defined as follows:
\begin{definition} \label{divdef} 
	Consider two distributions $S$ and $R$ in $ \mathcal{P}(\mathcal{Z})$.  A {\em divergence}  $D\colon \mathcal{P}^{2}(\mathcal{Z}) \rightarrow [0, \infty) $ between $S$ and $R$, denoted by  $D(S \| R)$ or $D(\mathbf{S}\| \mathbf{R})$, is a smooth function\footnote{A function is \emph{smooth} if it has partial derivatives of all orders.} of  $\mathbf{S}\in\Xi$ and $\mathbf{R}\in\Xi$ satisfying the following conditions: 
	\begin{enumerate} 
		\item $ D(S \| R) \geq 0$ for every $S, R \in \mathcal{P}(\mathcal{Z})$.  
		\item \label{div_cond2} $D(S \| R) =0$ if, and only if, $S=R$.
		\item Let $\mathbf{S}=\mathbf{R}+\bm{\varepsilon}$ for some $\bm{\varepsilon}=(\varepsilon_{1}, \ldots,\varepsilon_{k-1})^{\mathsf{T}}$. Then,
		\begin{equation}
		D(\mathbf{R} +\bm{ \varepsilon} \| \mathbf{R}   ) =\frac{1}{2} \sum_{i,j =1}^{k-1} g_{ij}(\mathbf{R} )  \varepsilon_{i} \varepsilon_{j} +O(\| \bm{\varepsilon} \|_{2}^{3})   \label{eq:divdef1}
		\end{equation}
		as $ \| \bm{\varepsilon} \|_{2} \rightarrow 0$
		for some  $(k-1) \times (k-1)$-dimensional positive-definite matrix $G(\mathbf{R})=[g_{ij}(\mathbf{R})]$ that depends on $\mathbf{R}$. In \eqref{eq:divdef1}, $\| \bm{\varepsilon} \|_{2}$ is the Euclidean norm of $\bm{\varepsilon}$.
        \item Let $R\in\mathcal{P}(\mathcal{Z})$, and let $\{S_n\}$ be a sequence of distributions in $\mathcal{P}(\mathcal{Z})$ that converges to a distribution $S$ on the boundary of $\mathcal{P}(\mathcal{Z})$. Then,
        \begin{equation}
        \varliminf_{n\to\infty} D(S_n \| R) > 0.
        \end{equation}
		\end{enumerate}	
\end{definition}
\begin{remark}
We follow the definition of divergence from the information geometry literature; see, e.g., \cite{AB216,CL86,CZPA209}. In particular, according to \cite[Def.~1.1]{AB216}, a divergence must satisfy the first three conditions in Definition~\ref{divdef}. Often, the behavior of divergence on the boundary of $\mathcal{P}(\mathcal{Z})$ is not specified. In Definition~\ref{divdef}, we add the fourth condition to treat the case of sequences of distributions $\{S_n\}$ that lie in $\mathcal{P}(\mathcal{Z})$ but converge to a distribution on the boundary of $\mathcal{P}(\mathcal{Z})$. 
\end{remark}
\begin{remark}
    Some divergences in the information theory literature, such as the total variation distance, do not qualify as divergences under Definition~\ref{divdef}, since they are not smooth functions of $\mathbf{S}$ and $\mathbf{R}$. Establishing the second-order asymptotics of divergence tests with such divergences would therefore require a different approach from the one adopted in this paper.
\end{remark}

By computing the partial derivatives of  $D(S \| R)$ with respect to the first variable $\mathbf{S}=(S_{1}, \ldots, S_{k-1})^{\mathsf{T}} $, it follows from the third condition in Definition~\ref{divdef} that 
\begin{equation}
D(S \| \br) 
= (\mathbf{S}-\mathbf{\br})^{T} \bm{A}_{D,\mathbf{\br}} ( \mathbf{S}-\mathbf{\br})+O(\|\mathbf{S}-\mathbf{\br}\|_{2}^{3}) \label{eq:tayldivA1}
\end{equation}
 as $ \| \mathbf{S}-\mathbf{\br} \|_{2} \rightarrow 0$, 
where $\bm{A}_{D, \mathbf{\br}} $ is the matrix associated with the divergence $D$ at $\mathbf{\br}$, which has components 
\begin{equation}
a_{ij} (\mathbf{\br})\triangleq \frac{1}{2} 	\left.	\frac{\partial^{2}}{\partial S_{i} \partial  S_{j}} D(S\| \br) \right|_{S=\br}.
\end{equation}
Based on $\bm{A}_{D,\mathbf{\br}}$, we can introduce the notion of an \emph{invariant divergence}.

\begin{definition}\label{invariancediv}
	Let $D$ be a divergence, and let $\br  \in \mathcal{P}(\mathcal{Z})$. Then, $D$ is said to be an \emph{invariant divergence} on $\mathcal{P}(\mathcal{Z})$ if the matrix associated with the divergence $D$ at $\mathbf{\br}$  is of the form $\bm{A}_{D, \mathbf{\br}}=\eta \bm{\Sigma}_{\mathbf{\br}}$ for a constant $\eta >0$ and a matrix $\bm{\Sigma}_{\mathbf{\br}}$ with components
	\begin{equation}
	\bm{\Sigma}_{ij}(\mathbf{\br})
	= 
	\begin{cases}
	\frac{1}{\br_{i}} +	\frac{1}{\br_{k}}, \quad  &  i=j \\
	\frac{1}{\br_{k}}, \quad & i \neq j. \\
	\end{cases}
	 \label{eq:covsigmain1}
	\end{equation}
\end{definition}

The notion of an invariant divergence is adapted from the notion of invariance of geometric structures in information geometry; see \cite{AB216}, \cite{CL86} for more details. Well-known divergences, such as the KL divergence, the JS divergence, the $f$-divergences, and the R\'enyi divergence, are invariant \cite{CZPA209}.

\subsection{Divergence Tests}

For a divergence $D$ and a threshold $r>0$, a \emph{divergence test} $\mc_{n}^{D}(r)$ for testing $H_{0}$ against the $H_{1}$ is defined as follows:
\begin{center}
	\begin{tabular}{ll} &Observe  two independent sequences $X^{n}, Y^{n}$: \\
		&if $D(\tp_{X^n} \| \tp_{Y^n}) < r$, then $H_0$ is accepted;\\
		& else $H_1$ is accepted.
	\end{tabular} 
\end{center} 
Here, $\tp_{X^{n}}$ and $\tp_{Y^{n}}$ are the types (empirical distributions) of the sequences $X^{n}$ and $Y^{n}$, respectively. This corresponds to the acceptance region
\begin{IEEEeqnarray}{lCl}
	\mathcal{A}^{D}_{n}(r) \triangleq	\left\lbrace  (x^n,y^n)\in\mathcal{Z}^n\times\mathcal{Z}^n \colon D(\tp_{x^n}\|\tp_{y^n})  <  r \right\rbrace.  
\end{IEEEeqnarray}
Let $D_{\textnormal{JS}}$ denote the JS divergence. For $D=D_{\textnormal{JS}}$, the divergence test $\mc_{n}^{D_{\textnormal{JS}}}$ specializes to the Gutman test.

\section{Robust Goodness-of-Fit Testing and the GLRT}
\label{sec:GLRT}
The Gutman test is the generalized likelihood ratio test (GLRT) for the two-sample testing problem. In the following, we recover this result by expressing the divergence test $\mc_{n}^{D_{\textnormal{JS}}}$ as a particular robust goodness-of-fit (GoF) test. Specifically, robust GoF testing is a binary hypothesis testing problem where, under hypothesis $H_0$, the sequence of observations $Z^n$ is distributed i.i.d.\ according to an unknown distribution $P$ that lies in the uncertainty class $\mathcal{C}\subset \bar{\mathcal{P}}(\mathcal{Z})$; under hypothesis $H_1$, the sequence of observations $Z^n$ is distributed i.i.d. according to an unknown distribution $Q$ that does  not lie in $\mathcal{C}$. The GLRT of this problem decides on $H_0$ if the test statistic
\begin{equation}
	D_{\textnormal{KL}}^{\textnormal{ROB}}(\tp_{Z^n} \| \mathcal{C})
	\triangleq
	\inf_{P \in \mathcal{C}} D_{\textnormal{KL}}(\tp_{Z^n} \| P) \label{eq:ROB_test_1D}
\end{equation}
is below a given threshold $r>0$, and it decides on $H_1$ otherwise. Here, $D_{\textnormal{KL}}$ denotes the KL divergence \cite[Def.~2.1]{PW25}. When $\mathcal{C}$ is a moment class characterized by $d$ linearly-independent functions on $\mathcal{Z}$ (cf.~\cite[Eq.~(110)]{HJT25}), the first- and second-order terms of this test are given by \cite[Sec.~VI]{HJT25}
\begin{subequations}
\begin{IEEEeqnarray}{lCl}
	\beta' & = & D_{\textnormal{KL}}(P^{*} \| Q) \label{eq:beta'_1D}\\
	\beta'' & = & -\sqrt{V_{\textnormal{KL}}(P^{*} \| Q)}
	\sqrt{\mathsf{Q}^{-1}_{\chi^{2}_{d}}(\epsilon)}\label{eq:beta''_1D}
\end{IEEEeqnarray}
\end{subequations}
where $\mathsf{Q}^{-1}_{\chi^{2}_{d}}(\cdot)$ denotes the inverse of the tail probability of the chi-square distribution with $d$ degrees of freedom,
\begin{equation}
V_{\textnormal{KL}}(P^*\|Q) \triangleq \sum_{i=1}^k P^*_i\left(\ln\frac{P^*_i}{Q_i}-D_{\textnormal{KL}}(P^*\|Q)\right)^2
\end{equation}
denotes the KL divergence variance, and
\begin{equation}
	P^{*} \triangleq \arg \min_{P \in \mathcal{C}} D_{\textnormal{KL}}(P \| Q).
\end{equation}

Two-sample testing directly relates to robust GoF testing by setting
\begin{equation}
	\mathcal{C} = \{(P_{1},P_{2})\in \mathcal{P}^{2}(\mathcal{Z}) \colon P_{1}=P_{2} \}. \label{eq:def_class}
\end{equation}
Indeed, with \eqref{eq:def_class}, two-sample testing can be formulated as
\begin{subequations}
\label{eqs:2sample_GoF}
 \begin{IEEEeqnarray}{llCl}
    H_0\colon & (X^n, Y^n)\sim (P_{1},P_{2}) & \in & \mathcal{C}  \\
    H_1\colon & (X^n, Y^n) \sim (P_{1},P_{2}) & \notin & \mathcal{C}. 
    \end{IEEEeqnarray}
    \end{subequations}
We next determine the GLRT of this problem.
\begin{proposition}\label{Thm:glrt_rob}
	Consider the two-sample testing problem \eqref{eqs:2sample_GoF} with $\mathcal{C}$ defined in \eqref{eq:def_class}. The test statistic \eqref{eq:ROB_test_1D} generalized to the bi-variate uncertainty class \eqref{eq:def_class} evaluates to 
\begin{IEEEeqnarray}{lCl}
	\IEEEeqnarraymulticol{3}{l}{D_{\textnormal{KL}}^{\text{ROB}} ((\tp_{X^{n}},\tp_{Y^{n}}) \|  \mathcal{C})}\nonumber\\
    \quad & \triangleq & \inf_{(P_1,P_2) \in \mathcal{C}} D_{\textnormal{KL}}((\tp_{X^{n}},\tp_{Y^{n}})  \|  (P_1,P_2)) \nonumber\\
	& = & 4n D_{\textnormal{JS}}(\tp_{X^{n}}\| \tp_{Y^{n}}).
\end{IEEEeqnarray}
Consequently, the Gutman test $\mc_{n}^{D_{\textnormal{JS}}}$ is the GLRT of the two-sample testing problem. 
\end{proposition}
\begin{proof}
	See Appendix~\ref{sec:append_glrt_rob}. 
	\end{proof}

Assuming that \cite[Sec.~VI]{HJT25} also applies to the bi-variate case with uncertainty class \eqref{eq:def_class}, \eqref{eq:beta'_1D}--\eqref{eq:beta''_1D} suggest that the first- and second-order terms of the Gutman test are given by 
\begin{subequations}
\begin{IEEEeqnarray}{lCl}
	\beta' &= & D_{\textnormal{KL}}( (P^{*}, P^{*})\| (P_1,P_2)) \label{eq:beta'_2D}\\
	\beta'' &=& -\sqrt{V_{\textnormal{KL}}( (P^{*}, P^{*}) \| (P_1,P_2))}
	\sqrt{\mathsf{Q}^{-1}_{\chi^{2}_{k-1}}(\epsilon)}\label{eq:beta''_2D}
\end{IEEEeqnarray}
\end{subequations}
where
\begin{equation}
P^* \triangleq \arg \min_{P\in\mathcal{P}(\mathcal{Z})} D_{\textnormal{KL}}((P,P)\|(P_1,P_2)).\label{eq:min_Bhatt1}
\end{equation}
The minimum in \eqref{eq:min_Bhatt1} exists because the KL divergence is a continuous and convex function on the compact set $\bar{\mathcal{P}}(\mathcal{Z})$, and because $P_1, P_2 \in \mathcal{P}(\mathcal{Z})$.
By the chain rule of the KL divergence and the KL divergence variance, \eqref{eq:beta'_2D}--\eqref{eq:beta''_2D} can be written as
\begin{subequations}
\begin{IEEEeqnarray}{lCl}
\beta' &= & D_{\text{KL}}(P^{*}\| P_1) + D_{\text{KL}}(P^{*}\| P_2)	 \label{eq:beta'_2D_2}\\*
	\beta'' & = & -\sqrt{V_{\textnormal{KL}}(P^{*} \| P_1) + V_{\textnormal{KL}}(P^{*} \| P_2) }\sqrt{\mathsf{Q}^{-1}_{\chi^{2}_{k-1}}(\epsilon)}.\label{eq:beta''_2D_2}\IEEEeqnarraynumspace
\end{IEEEeqnarray}
\end{subequations}
It can be further shown that (see, e.g., \cite[Lemma~2]{Li_Universal})
\begin{IEEEeqnarray}{lCl}
\beta' & = & \min_{P\in\bar{\mathcal{P}}(\mathcal{Z})}\bigl\{D_{\text{KL}}(P\| P_1) + D_{\text{KL}}(P\| P_2)\bigr\} \nonumber\\
& = & 2 D_{\textnormal{B}}(P_1,P_2) \label{eq:twice_Bhatt}
\end{IEEEeqnarray}
and the minimum is achieved for
\begin{equation}
P^*(z) = \frac{\sqrt{P_{1}(z)P_{2}(z)}}{\sum_{z'\in\mathcal{Z}} \sqrt{P_{1}(z')P_{2}(z')}}, \quad z\in\mathcal{Z}.\label{eq:min_Bhatt1_2}
\end{equation}
In \eqref{eq:twice_Bhatt}, $D_{\textnormal{B}}(P_{1},P_{2})\triangleq- \ln \bigl(\sum_{z\in\mathcal{Z}} \sqrt{P_{1}(z)P_{2}(z)}\bigr)$ denotes the Bhattacharyya distance. 
The following theorem demonstrates that \eqref{eq:beta'_2D_2}--\eqref{eq:beta''_2D_2} characterize indeed the second-order asymptotics of the Gutman test:
\begin{theorem}\label{Thm:GLRT}
	Consider the two-sample testing problem described in Section~\ref{sec:setting}. Then, the Gutman test $\mc_{n}^{D_{\textnormal{JS}}}$ satisfies
	\begin{IEEEeqnarray}{lCl}
		\sup_{r_n\colon \alpha_n(\mc_{n}^{D_{\textnormal{JS}}}(r_n))\leq \epsilon} -\frac{\ln	\beta_n(\mc_{n}^{D_{\textnormal{JS}}}(r_n))}{n} & = &   2 D_{\textnormal{B}}(P_{1},P_{2}) \nonumber\\
		 \IEEEeqnarraymulticol{3}{r}{{}-\sqrt{\frac{V_{\text{KL}}(P^{*}\|P_1) + V_{\text{KL}}(P^{*}\|P_2) }{n}} \sqrt{\mathsf{Q}^{-1}_{\chi^{2}_{k-1}}(\epsilon)}+o\biggl(\frac{1}{\sqrt{n}}\biggr)}\nonumber\\\label{eq:second_glrt}
	\end{IEEEeqnarray}
	where $P^{*}$ is given in \eqref{eq:min_Bhatt1_2} and $ \mathsf{Q}^{-1}_{\chi^{2}_{k-1}}(\cdot)$ denotes the inverse of the tail probability of the chi-square distribution $\chi^{2}_{k-1}$ with $k-1$ degrees of freedom.
\end{theorem}
\begin{IEEEproof}
Note that the JS divergence is invariant. The second-order asymptotics of the Gutman test $\mc_{n}^{D_{\textnormal{JS}}}$ is thus a special case of Theorem~\ref{Thm:unknownPdivtest}, which characterizes the second-order asymptotics of the divergence test $\mc_{n}^{D}$ for the class of invariant divergences $D$.
\end{IEEEproof} 

\section{Asymptotics of Divergence Tests\\ with Invariant Divergences}
\label{sec:main_results}
In this section, we discuss the divergence test $\mc_{n}^{D}$ for the class of invariant divergences, as defined in Definition~\ref{invariancediv}. As mentioned above, many well-known divergences, including the $f$-divergences and the Renyi divergence, belong to this class. In particular, the JS divergence is an invariant divergence with constant $\eta=\frac{1}{8}$. The second-order asymptotics of such test is characterized in the following theorem.

\begin{theorem}\label{Thm:unknownPdivtest}
	Consider the two-sample testing problem described in Section~\ref{sec:setting}. Let $D$ be an invariant divergence, as defined in Definition~\ref{invariancediv}. Then, the divergence test $\textsf{T}_n^{D}$ satisfies
	\begin{IEEEeqnarray}{lCl}
		\sup_{r_n\colon \alpha_n(\mc_{n}^{D}(r_n))\leq \epsilon} -\frac{\ln	\beta_n(\mc_{n}^{D}(r_n))}{n} & = &   2 D_{\textnormal{B}}(P_{1},P_{2}) \nonumber\\
		 \IEEEeqnarraymulticol{3}{r}{{}-\sqrt{\frac{V_{\text{KL}}(P^{*}\|P_1) + V_{\text{KL}}(P^{*}\|P_2) }{n}} \sqrt{\mathsf{Q}^{-1}_{\chi^{2}_{k-1}}(\epsilon)}+o\biggl(\frac{1}{\sqrt{n}}\biggr)}\nonumber\\\label{eq:second_invariant}
	\end{IEEEeqnarray}
	where $P^{*}$ is given in \eqref{eq:min_Bhatt1_2}  and $ \mathsf{Q}^{-1}_{\chi^{2}_{k-1}}(\cdot)$ denotes the inverse of the tail probability of the chi-square distribution $\chi^{2}_{k-1}$ with $k-1$ degrees of freedom. 
\end{theorem}
\begin{IEEEproof}
	For an outline of the proof, see Section~\ref{sub:outline}. A full proof is given in Appendix~\ref{sec:proof}.
\end{IEEEproof}

The Gutman test compares the empirical distributions of the sequences $X^n$ and $Y^n$ and rejects the null hypothesis when their difference, measured by the JS divergence, exceeds a threshold. Theorem~\ref{Thm:unknownPdivtest} shows that, whether the empirical distributions $\tp_{X^n}$ and $\tp_{Y^n}$ are compared using the JS divergence, the KL divergence, or any other invariant divergence, the first- and second-order terms $\beta'$ and $\beta''$ remain unchanged.

\subsection{Proof Outline of Theorem~\ref{Thm:unknownPdivtest}}
\label{sub:outline}
The achievability and converse part hinge on the convergence of $\frac{n}{2}D (\tp_{X^{n}} \| \tp_{Y^{n}})$ to a generalized chi-square random variable, defined as follows:
\begin{definition}
    We shall say that the random variable $\chi^{2}_{\bm{w},m}$ is a generalized chi-square random variable with vector parameter $\bm{w}=(w_1,\ldots,w_m)$ and $m$ degrees of freedom if it has the same distribution as the random variable
    \begin{equation}
        \xi = \sum_{i=1}^m w_i \Upsilon_i
    \end{equation}
    where $(\Upsilon_1,\ldots,\Upsilon_m)$ are i.i.d.\ chi-square random variables with $1$ degree of freedom.
\end{definition}

The following lemma characterizes the convergence of $\frac{n}{2}D(\tp_{X^{n}} \| \tp_{Y^{n}})$ to a generalized chi-square random variable. In the special case where $D$ is the JS divergence, this result reduces essentially to \cite[Th.~6]{Unnikrishnan16}.
\begin{lemma} \label{cov_lemma} Let $X^n$ and $Y^n$ be independent sequences of i.i.d.\ random variables distributed according to $P$, and let $D$ be a divergence. Further let $\bm{\lambda}=(\lambda_{1}, \ldots, \lambda_{k-1})^{\mathsf{T}}$ be a vector that contains the eigenvalues of the matrix $ \bm{\Sigma}_{\mathbf{P}}^{-1/2}\bm{A}_{D,\mathbf{P}} \bm{\Sigma}_{\mathbf{P}}^{-1/2}$, where $\bm{A}_{D,\mathbf{P}}$ is the matrix associated with the divergence $D$ at $P$ (cf.~\eqref{eq:tayldivA1}) and the matrix $\bm{\Sigma}_{\mathbf{P}}$ is defined in \eqref{eq:covsigmain1}. Then,
		\begin{equation}
			P^{n} \left( \frac{n}{2}D( \tp_{X^{n}}\| \tp_{Y^{n}})  \geq c \right) =  \mathsf{Q}_{\chi^{2}_{\bm{\lambda},k-1}}(c) + O(\delta_{n}) \label{eq:ratediv2d1c}
		\end{equation}
    		for all $c >0$ and some positive sequence  $\{\delta_{n}\}$ that is independent of $c$ and satisfies  $\lim_{n \rightarrow \infty} \delta_{n}=0$, and where the $O(\delta_{n})$-term is uniform in $c$. In \eqref{eq:ratediv2d1c}, $ \mathsf{Q}_{\chi^{2}_{\bm{\lambda},k-1}}(\cdot)$ is the tail probability of the generalized chi-square distribution $\chi^{2}_{\bm{\lambda},k-1}$ with vector parameter
		$\bm{\lambda}$ and $k-1$ degrees of freedom.
	\end{lemma}
	\begin{IEEEproof}
		See Appendix~\ref{sec:appendA_con_div}. 
	\end{IEEEproof}
    
 For any invariant divergence, we have that $\bm{A}_{D, \mathbf{P}}=\eta \bm{\Sigma}_{\mathbf{P}}$ for some constant $\eta>0$ (cf.~Definition~\ref{invariancediv}). It follows that the eigenvalues of the matrix $ \bm{\Sigma}_{\mathbf{P}}^{-1/2}\bm{A}_{D,\mathbf{P}} \bm{\Sigma}_{\mathbf{P}}^{-1/2}$ are $\bm{\lambda}=(\eta,\ldots,\eta)$. This implies that $\mathsf{Q}_{\chi^{2}_{\bm{\lambda},k-1}}(c) =  \mathsf{Q}_{\chi^{2}_{k-1}}(c/\eta)$ \cite[Eq.~(120)]{HJT25}, so for an invariant divergence $D$, \eqref{eq:ratediv2d1c} can be written as
 \begin{equation}
    P^{n} \left( \frac{n}{2}D( \tp_{X^{n}}\| \tp_{Y^{n}})  \geq c \right) =  \mathsf{Q}_{\chi^{2}_{k-1}}(c/\eta) + O(\delta_{n}). \label{eq:ratediv_inv}
\end{equation}
From \eqref{eq:ratediv_inv}, it can be shown that the smallest threshold value $r_n$ for which $\alpha_{n}(\mc_{n}^{D}(r_{n})) \leq \epsilon$, denoted as $r_{n,\epsilon}^{D}$, satisfies
\begin{equation}
		r^{D}_{n, \epsilon} = \frac{2\eta }{n} \mathsf{Q}^{-1}_{\chi^{2}_{k-1}} ( \epsilon)+  O \left( \frac{\delta_{n}}{n}  \right). \label{eq:sec6rbdi1}
	\end{equation}

We next bound the type-II error for $r^{D}_{n, \epsilon}$. To upper-bound $\beta_n(\textsf{T}_n^{D}(r^{D}_{n, \epsilon}))$, we use the method of types \cite[Th.~11.1.4]{ATCB} to obtain
\begin{IEEEeqnarray}{lCl}
		\IEEEeqnarraymulticol{3}{l}{\beta_n(\textsf{T}_n^{D}(r^{D}_{n, \epsilon}))} \nonumber\\
         & \leq & \sum_{(P', P'') \in 	\mathcal{B}^{D}(r^{D}_{n, \epsilon})\cap (\mathcal{P}_n \times \mathcal{P}_n)} e^{-n D_{\textnormal{KL}}(P'\| P_1)-n D_{\textnormal{KL}}(P'' \| P_2)}\IEEEeqnarraynumspace      
	\end{IEEEeqnarray}
where $\mathcal{P}_n$ denotes the set of types with denominator $n$ and
	\begin{IEEEeqnarray}{lCl}
		\mathcal{B}_{D}(r) \triangleq \left\lbrace  (T,R) \in \mathcal{P}^2(\mathcal{Z}) \colon D(T,R)  < r \right\rbrace. \label{eq:B_D(r)}
	\end{IEEEeqnarray}
 We then lower-bound $D_{\text{KL}}(P' \| P_{1})+D_{\text{KL}}(P'' \| P_{2})$ over $(P', P'') \in \mathcal{B}_{D}(r^{D}_{n, \epsilon})$. To this end, we first note that $r^{D}_{n, \epsilon} = O(1/n)$ and, hence, also $D(P',P'') = O(1/n)$ for every $(P',P'')\in\mathcal{B}_{D}(r^{D}_{n, \epsilon})$. It follows from \cite[Lemma~2]{HJT25} that $\|\mathbf{P}'-\mathbf{P}''\|_2 = O(1/\sqrt{n})$, and hence also $\|\mathbf{P}'-\bar{\mathbf{P}}\|_2+\|\mathbf{P}''-\bar{\mathbf{P}}\|_2=O(1/\sqrt{n})$ for $\bar{P} = \frac{1}{2} P' + \frac{1}{2} P''$. We can thus perform a Taylor-series approximation of $D_{\text{KL}}(P' \| P_{1})+D_{\text{KL}}(P'' \| P_{2})$ around $(\bar{P},\bar{P})$, which we minimize over $(P',P'')\in \mathcal{B}_{D}(r^{D}_{n, \epsilon})$ by showing that there exists an $r_n'>0$ such that any pair $(P',P'')\in\mathcal{B}_{D}(r^{D}_{n, \epsilon})$ lies in the set
 \begin{IEEEeqnarray}{lCl}
\mathcal{A}_{\bm{\Sigma}_{\bar{\mathbf{P}}}}(r_n') & \triangleq & \bigl\{ (T,R)\in\mathcal{P}^2(\mathcal{Z}) \colon 	(\mathbf{T}-\bar{\mathbf{P}})^{T}\bm{\Sigma}_{\bar{\mathbf{P}}}(\mathbf{T}-\bar{\mathbf{P}}) \notag \\ 
        & & \quad\quad\qquad {} +(\mathbf{R}-\bar{\mathbf{P}})^{T}\bm{\Sigma}_{\bar{\mathbf{P}}}(\mathbf{R}-\bar{\mathbf{P}}) \leq r_n'\bigr\} \IEEEeqnarraynumspace \label{eq:A_Sigma}
 \end{IEEEeqnarray}
and by then minimizing the Taylor-series approximation of $D_{\text{KL}}(P' \| P_{1})+D_{\text{KL}}(P'' \| P_{2})$ over $(P',P'')\in\mathcal{A}_{\bm{\Sigma}_{\bar{\mathbf{P}}}}(r_n')$.

To lower-bound $\beta_n(\textsf{T}_n^{D}(r^{D}_{n, \epsilon}))$, we first show that one can find a threshold value $\tilde{r}_n$ that has the same order as $r^{D}_{n, \epsilon}$ and satisfies $\mathcal{A}_{\bm{\Sigma}_{\mathbf{P}^*}}(\tilde{r}_n) \subseteq \mathcal{B}_D(r^{D}_{n, \epsilon})$, where $\mathcal{A}_{\bm{\Sigma}_{\mathbf{P}^*}}$ is as in \eqref{eq:A_Sigma} but with $\bar{\mathbf{P}}$ replaced by $\mathbf{P}^*$, and $P^{*}$ is given in \eqref{eq:min_Bhatt1_2}. We then use again the method of types \cite[Th.~11.1.4]{ATCB} to obtain
\begin{IEEEeqnarray}{lCl}
\IEEEeqnarraymulticol{3}{l}{\beta_{n}(\mc_{n}^{D}(r^{D}_{n,\epsilon})) } \nonumber\\
\quad & \geq &\frac{1}{(n+1)^{2|\mathcal{Z}|}} e^{-nD_{\textnormal{KL}}(\tp'_{n} \| P_{1}) - nD_{\textnormal{KL}}(\tp''_{n} \| P_{2})} 
\end{IEEEeqnarray}
for some types $(\tp'_{n},\tp''_{n})$ that lie in $\mathcal{A}_{\bm{\Sigma}_{\mathbf{P}^*}}(\tilde{r}_n)$. Since $r_{n,\epsilon}^{D}$, and hence also $\tilde{r}_n$, are of order $1/n$, it follows from \cite[Lemma~2]{HJT25} that $\|\mathbf{P}_n'-\mathbf{P}^*\|_2 + \|\mathbf{P}_n''-\mathbf{P}^*\|_2 = O(1/\sqrt{n})$. We can thus perform a Taylor-series approximation of $D_{\textnormal{KL}}(\tp'_{n} \| P_{1})+D_{\textnormal{KL}}(\tp''_{n} \| P_{2})$ around $(P^{*},P^{*})$. We conclude the proof by demonstrating that there exists a pair of types $(\tp'_n, \tp''_n)$ in $\mathcal{A}_{\bm{\Sigma}_{\mathbf{P}^*}}(\tilde{r}_n)$ for which the Taylor-series approximation of $D_{\textnormal{KL}}(\tp'_{n} \| P_{1})+D_{\textnormal{KL}}(\tp''_{n} \| P_{2})$ is close to the Taylor-series approximation of $D_{\textnormal{KL}}(P' \| P_{1})+D_{\textnormal{KL}}( P''\| P_{2})$ minimized over $(P',P'') \in\mathcal{B}_{D}(r^{D}_{n,\epsilon})$.
        
\section{Asymptotics of Divergence Tests\\ with Non-Invariant Divergences}
\label{sec:non-invariant}

\comment{In \cite{HJT25}, we proposed a divergence test for the GoF problem and characterized its second-order asymptotics for general divergences. Specifically, GoF testing is a binary hypothesis testing problem where, under hypothesis $H_0$, the sequence of observations $Z^n$ is distributed i.i.d.\ according to a known distribution $P$, and under hypothesis $H_1$, the sequence of observations $Z^n$ is distributed i.i.d.\ according to an unknown distribution $Q$ different from $P$. Thus, GoF testing is a special case of robust GoF testing when the uncertainty class $\mathcal{C}$ only contains one distribution $P$. The GLRT of this problem is the so-called Hoeffding test \cite{H65}, which decides on $H_0$ if $D_{\textnormal{KL}}(\tp_{Z^n}\| P)$ is below a given threshold $r>0$, and it decides on $H_1$ otherwise. The divergence test for this problem replaces the KL divergence $D_{\textnormal{KL}}$ by an arbitrary divergence $D$.}

In \cite{HJT25}, we proposed a divergence test for the GoF problem and characterized its second-order asymptotics for general divergences. Specifically, GoF testing is a special case of robust GoF testing, introduced in Section~\ref{sec:GLRT}, when the uncertainty class $\mathcal{C}$ contains only one distribution $P$. The GLRT of this problem is the so-called Hoeffding test \cite{H65}, which decides on $H_0$ if $D_{\textnormal{KL}}(\tp_{Z^n}\| P)$ is below a given threshold $r>0$, and it decides on $H_1$ otherwise. The divergence test for this problem replaces the KL divergence $D_{\textnormal{KL}}$ by an arbitrary divergence $D$.

Similar to Theorem~\ref{Thm:unknownPdivtest}, we showed in \cite[Cor.~1]{HJT25} that  divergence tests with invariant divergences achieve the same second-order asymptotics as the GLRT. In contrast, divergence tests with non-invariant divergences achieve the same first-order term as the GLRT, but the second-order term may be strictly larger than that of the GLRT for some alternative distributions $Q$. Thus, there are alternative distributions $Q$ for which a divergence test with a non-invariant divergence has a better second-order performance than the GLRT.

This raises the question whether divergence tests with non-invariant divergences may also achieve a better second-order performance than the GLRT for the two-sample testing problem. However, a characterization of the second-order asymptotics of the divergence test $\mc_{n}^{D}$ for general divergences remains open. The reason is that, for non-invariant divergences, the tail probability of the test statistic $D(\tp_{X^n}\| \tp_{Y^n})$ is approximated by $\mathsf{Q}_{\chi^{2}_{\bm{\lambda},k-1}}$, which depends on the vector of eigenvalues $\bm{\lambda}$ of the matrix $ \bm{\Sigma}_{\mathbf{P}}^{-1/2}\bm{A}_{D,\mathbf{P}} \bm{\Sigma}_{\mathbf{P}}^{-1/2}$ (cf.~Lemma~\ref{cov_lemma}). Since $\bm{\lambda}$ is a function of $P$, which is not available to the test, finding a threshold $r_n>0$ for which $\alpha_{n}(\mc_{n}^{D}(r_{n})) \leq \epsilon$, and bounding then the type-II error for this threshold, becomes considerably more challenging. Nevertheless, the first-order asymptotics of the divergence test $\mc_{n}^{D}$ for general divergences can be characterized:

\begin{proposition}\label{prop:gendiv}
Consider the two-sample testing problem described in Section~\ref{sec:setting}. Let $D$ be an arbitrary divergence, as defined in Definition~\ref{divdef}. Then, there exists a sequence of threshold values $\{r_n\}$ such that the divergence test $\textsf{T}_n^{D}$ satisfies
\begin{subequations}
	\begin{IEEEeqnarray}{rCl}
        \lim_{n\to\infty} \alpha_n(\mc_{n}^{D}(r_n)) & = & 0 \label{eq:prop:gendiv_a}\\
		\lim_{n\to\infty} -\frac{\ln\beta_n(\mc_{n}^{D}(r_n))}{n} & = &   2 D_{\textnormal{B}}(P_{1},P_{2}). \IEEEeqnarraynumspace\label{eq:prop_gendiv_b}
	\end{IEEEeqnarray}
\end{subequations}
\end{proposition}
\begin{IEEEproof}
See Appendix~\ref{Append:bla}.
\end{IEEEproof}

Proposition~\ref{prop:gendiv} demonstrates that, irrespective of the divergence $D$, the divergence test $\textsf{T}_n^{D}$ achieves the same first-order term $\beta'$ as the GLRT $\mc_{n}^{D_{\textnormal{JS}}}$. In fact, it was shown in \cite[Th.~10]{Shengyu21} that no two-sample test with type-I error bounded by $\epsilon\in(0,1)$ can achieve a type-II first-order term $\beta'$ that exceeds $2 D_{\textnormal{B}}(P_1,P_2)$. So the divergence test $\textsf{T}_n^{D}$ is first-order optimal, irrespective of the divergence $D$.

Zhu \emph{et al.} \cite{Shengyu21} showed that the MMD-based test achieves the optimal first-order term $2 D_{\textnormal{B}}(P_1,P_2)$ even when the observations take values in an arbitrary Polish space. For discrete alphabets $\mathcal{Z}$, the MMD can be shown to be a non-invariant divergence. In this setting, the first-order optimality of the MMD-based test follows as a special case of Proposition~\ref{prop:gendiv}. However, extending our second-order asymptotic analysis beyond finite alphabets appears challenging, even within the class of invariant divergences. Indeed, as seen in Theorem~\ref{Thm:unknownPdivtest}, the second-order term $\beta''$ depends on the alphabet size $k$ through the term $\mathsf{Q}^{-1}_{\chi^{2}_{k-1}}(\epsilon)$, which is unbounded in $k$.

\section*{Acknowledgment}
Jithin Ravi thanks Srikrishna Bhashyam (IIT Madras) for fruitful discussions during the initial phase of this work.

\newpage

\appendices

\section{Proof of Proposition~\ref{Thm:glrt_rob}}
\label{sec:append_glrt_rob}

Consider a standard composite binary hypothesis testing problem of the form
\begin{equation}
H_0:\ \theta\in\Theta_0 \qquad\text{vs.}\qquad H_1:\ \theta\in \Theta \setminus \Theta_0
\end{equation}
where $\Theta$ is the \emph{full parameter space} and $\Theta_0$ is the \emph{restricted space under the null hypothesis}.
Let $L(\theta | z)$ denote the likelihood function, i.e., $L(\theta | z)$ is  conditional probability mass function (PMF) of the observation $Z=z$ given the parameter $\theta$, and define 
\begin{subequations}
\begin{IEEEeqnarray}{lCl}
L(\hat\theta_0 | z) &\triangleq & \sup_{\theta\in\Theta_0}L(\theta |z)\\
L(\hat\theta | z) & \triangleq & \sup_{\theta\in\Theta}L(\theta | z).
\end{IEEEeqnarray}
\end{subequations}
Furthermore, define the generalized likelihood ratio as
\begin{equation}
\Lambda(z)
\triangleq\frac{L(\hat\theta_0 | z)}{L(\hat\theta | z)}.
\end{equation}
The GLRT statistic is then given by \cite[Ch.~12]{LR205}
\begin{equation}
T(z)=-2\ln\Lambda(z).
\end{equation}

Specialized to the two-sample testing problem, the full parameter space $\Theta$ and restricted parameter space $\Theta_0$ become
\begin{subequations}
\begin{IEEEeqnarray}{lCl}
\Theta & = & \bar{\mathcal{P}}^{2}(\mathcal{Z})\\
\Theta_0 & = & \mathcal{C}
\end{IEEEeqnarray}
\end{subequations}
where $ \mathcal{C}$ is defined in \eqref{eq:def_class}. Furthermore, the observation is given by $Z=(X^n,Y^n)$. We next compute the GLRT statistic of the two-sample testing problem. Indeed, $X^n$ and $Y^n$ are independent sequences of i.i.d.\ random variables distributed according to $P_1$ and $P_2$, respectively. It can then be shown that
\begin{IEEEeqnarray}{lCl}
	\ln L(\hat{\theta}|z) &	= &  \sum_{i=1}^k n \tp_{x^{n}}(a_i) \ln  \tp_{x^{n}}(a_i) \nonumber\\
    & & {} +  \sum_{i=1}^k n \tp_{y^{n}}(a_i) \ln \tp_{y^{n}}(a_i)
\end{IEEEeqnarray}
and
\begin{IEEEeqnarray}{lCl}
	\IEEEeqnarraymulticol{3}{l}{\ln L(\hat\theta_0|z)} \nonumber \\
	& = & \sum_{i=1}^k n\left(   \tp_{x^{n}}(a_i) +\tp_{x^{n}}(a_i) \right)  \ln \left(\frac{1}{2}\tp_{x^n}(a_i)+\frac{1}{2}\tp_{y^n}(a_i) \right). \nonumber\\
\end{IEEEeqnarray}
Substituting these expressions into the GLRT statistic $T(x^n,y^n)$, we obtain that
\begin{IEEEeqnarray}{lCl}
	T(x^n,y^n) &= & -2 \ln\Lambda(x^n,y^n) \nonumber\\
    & = & -2\ln L(\hat{\theta}|z) - 2 \ln L(\hat{\theta_{0}}|z) \notag \\
    & = & 2 n\sum_{i=1}^k \tp_{x^n}(a_i)\ln \frac{\tp_{x^n}(a_i)}{\frac{1}{2}\tp_{x^n}(a_i)+\frac{1}{2}\tp_{y^n}(a_i)}\nonumber\\
    & & {} + 2 n\sum_{i=1}^k \tp_{y^n}(a_i)\ln \frac{\tp_{x^n}(a_i)}{\frac{1}{2}\tp_{x^n}(a_i)+\frac{1}{2}\tp_{y^n}(a_i)} \nonumber\\
	&= & 4n D_{\textnormal{JS}}(\tp_{x^{n}}\| \tp_{y^{n}}) \label{eq:GLRT_unknownP1}
\end{IEEEeqnarray}
where in the last step, we used that the Jensen–Shannon divergence is given by
\begin{IEEEeqnarray}{lCl}
\IEEEeqnarraymulticol{3}{l}{D_{\textnormal{JS}}(T \| R )} \nonumber\\
\quad & = &  \frac{1}{2} D_{\textnormal{KL}} \biggl( T \biggm\|  \frac{T+R}{2}\biggr)  +\frac{1}{2} D_{\textnormal{KL}} \biggl( R  \biggm\|  \frac{T+R}{2}\biggr). \IEEEeqnarraynumspace
\end{IEEEeqnarray}

Next, we consider the robust GoF test statistic 
\begin{IEEEeqnarray}{lCl}
    \IEEEeqnarraymulticol{3}{l}{D_{\textnormal{KL}}^{\text{ROB}} ((\tp_{X^{n}},\tp_{Y^{n}}) \|  \mathcal{C}) }\nonumber\\
    \quad & = & \inf_{(P,P) \in \mathcal{C}} D_{\textnormal{KL}}((\tp_{X^{n}},\tp_{Y^{n}})  \|  (P,P)). \label{eq:glrt_inf}
\end{IEEEeqnarray}
Since
\begin{IEEEeqnarray}{lCl}
\IEEEeqnarraymulticol{3}{l}{D_{\textnormal{KL}}((\tp_{X^{n}},\tp_{Y^{n}})  \|  (P,P))} \nonumber\\
\quad & = & D_{\textnormal{KL}}(\tp_{X^{n}} \|P)+ D_{\textnormal{KL}}(\tp_{Y^{n}} \|P)
\end{IEEEeqnarray}
 the infimum in \eqref{eq:glrt_inf} is equivalent to
 \begin{equation}
   \inf_{P \in \mathcal{P}(\mathcal{Z})} \left\{ D_{\textnormal{KL}}(\tp_{X^{n}} \|P)+ D_{\textnormal{KL}}(\tp_{Y^{n}} \|P)\right\}.	
 \end{equation}
By Lagrange optimization, it can be shown that the infimum is attained at $P=\frac{1}{2}\tp_{X^{n}}+\frac{1}{2}\tp_{Y^{n}}$. Substituting this into \eqref{eq:glrt_inf}, we obtain that
\begin{equation}
  D_{\textnormal{KL}}^{\text{ROB}} ((\tp_{X^{n}},\tp_{Y^{n}}) \|  \mathcal{C}) =  4n D_{\textnormal{JS}}(\tp_{X^{n}}\| \tp_{Y^{n}})
\end{equation}
which coincides with the GLRT test statistic. 

\section{Proof of Lemma \ref{cov_lemma}}
\label{sec:appendA_con_div}
A Taylor-series expansion of $D(T\|R)$ around the point $(P,P)$ yields that
\begin{IEEEeqnarray}{lCl}
    \IEEEeqnarraymulticol{3}{l}{D(T\|R)} \nonumber\\
    \quad & = & \frac{1}{2}\sum_{i=1}^{k-1}\sum_{j=1}^{k-1}	\frac{\partial^{2} D(P\| P) }{\partial T_{i} \partial  T_{j}} (T_{i}-P_{i})(T_{j}-P_{j}) \notag \\
	& & {} +\frac{1}{2}\sum_{i=1}^{k-1}\sum_{j=1}^{k-1}	\frac{\partial^{2} D(P\| P) }{\partial R_{i} \partial  R_{j}} (R_{i}-P_{i})(R_{j}-P_{j}) \notag \\
	& & {}	-\frac{1}{2}\sum_{i=1}^{k-1}\sum_{j=1}^{k-1}	2\frac{\partial^{2} D(P\| P) }{\partial R_{i} \partial  T_{j}} (R_{i}-P_{i})(T_{j}-P_{j}) \notag \\
	& & {} +O(\|\mathbf{T}-\mathbf{P}\|_{2}^{3}) +O(\|\mathbf{R}-\mathbf{P}\|_{2}^{3}).
	\label{eq:tayldivA}
\end{IEEEeqnarray}
Since any divergence $D(T\|R)$ satisfies at $T=R=P$ that
\begin{IEEEeqnarray}{lCl}
	\frac{\partial^{2} D(P\| P) }{\partial T_{i} \partial  T_{j}} 
	& = & 	\frac{\partial^{2} D(P\| P) }{\partial R_{i} \partial  R_{j}}  \notag \\
	& = & -	\frac{\partial^{2} D(P\| P) }{\partial T_{i} \partial  R_{j}}, \quad  i,j=1,\ldots, k-1 \IEEEeqnarraynumspace
\end{IEEEeqnarray}
and since each component of the matrix $\bm{A}_{D, \mathbf{\br}} $ associated with the divergence $D$ at $\mathbf{\br}$ is given by
\begin{equation}
	\bm{A}_{D,\mathbf{P}}(i,j)=  \frac{1}{2} 	\left.	\frac{\partial^{2} D(T\| \br) }{\partial T_{i} \partial  T_{j}} \right|_{T=R=P} \label{eq:matrixa}
\end{equation}
the Taylor-series approximation \eqref{eq:tayldivA} can be written as
\begin{IEEEeqnarray}{lCl}
    \IEEEeqnarraymulticol{3}{l}{D(T\|R)} \nonumber\\
    \quad & = & \begin{bmatrix}
		\mathbf{T}-\mathbf{P} & \mathbf{R}-\mathbf{P} 
	\end{bmatrix}^{T}  \begin{bmatrix}
		\bm{A}_{D,\mathbf{P}}& 	-\bm{A}_{D,\mathbf{P}} \\
		-\bm{A}_{D,\mathbf{P}}& 	\bm{A}_{D,\mathbf{P}}
	\end{bmatrix}   \begin{bmatrix}
		\mathbf{T}-\mathbf{P} \\ \mathbf{R}-\mathbf{P} 
	\end{bmatrix}  \notag \\
	& & {} +O(\|\mathbf{T}-\mathbf{P}\|_{2}^{3}) +O(\|\mathbf{R}-\mathbf{P}\|_{2}^{3}). \label{eq:(47)}
\end{IEEEeqnarray}

\comment{Consider the eigenvalue decomposition of the matrix $\bm{A}_{D,\mathbf{P}}=V \Lambda V^{T}$, where $V$ is an orthogonal matrix containing the eigenvectors of $\bm{A}_{D,\mathbf{P}}$, and $\Lambda$ is a diagonal matrix containing the eigenvalues of $\bm{A}_{D,\mathbf{P}}$. It follows that
\begin{IEEEeqnarray}{lCl}
	\begin{bmatrix}
		\bm{A}_{D,\mathbf{P}}& 	-\bm{A}_{D,\mathbf{P}} \\
		-\bm{A}_{D,\mathbf{P}}& 	\bm{A}_{D,\mathbf{P}}
	\end{bmatrix}
	& =& \begin{bmatrix}
		V& 	V\\
		V& 	-V
	\end{bmatrix} \begin{bmatrix}
		\bm{0}& 	\bm{0}\\
		\bm{0}& 	\Lambda
	\end{bmatrix} \begin{bmatrix}
		V& 	V\\
		V& 	-V
	\end{bmatrix}. \IEEEeqnarraynumspace
\end{IEEEeqnarray}}
To shorten notation, let $	\mathbf{a}=	\mathbf{T}-\mathbf{P}$ and $\mathbf{b}=	\mathbf{R}-\mathbf{P}$. We can then write \eqref{eq:(47)} as 
\begin{IEEEeqnarray}{lCl}
	\IEEEeqnarraymulticol{3}{l}{D(T\|R)} \nonumber \\ 
    &= & (\mathbf{a}-\mathbf{b})^{T} 	\bm{A}_{D,\mathbf{P}}(\mathbf{a}-\mathbf{b}) +O(\|\mathbf{T}-\mathbf{P}\|_{2}^{3}) +O(\|\mathbf{R}-\mathbf{P}\|_{2}^{3}) \notag \\
	&=&	(\mathbf{T}-\mathbf{R})^{T}\bm{A}_{D,\mathbf{P}}(\mathbf{T}-\mathbf{R}) \nonumber\\
    & & {} +O(\|\mathbf{T}-\mathbf{P}\|_{2}^{3}) +O(\|\mathbf{R}-\mathbf{P}\|_{2}^{3}). \label{eq:div_quad}
\end{IEEEeqnarray}

Now consider two independent type sequences $ \mathbf{\tp}_{X^{n}}$ and  $\mathbf{\tp}_{Y^{n}}$  under the null hypothesis $(P_1,P_2)=(P,P)$.  As $n \to \infty$,
\begin{eqnarray}
	\sqrt{n}\left(  \mathbf{\tp}_{X^{n}} -\mathbf{P} \right) &\xrightarrow{d} & \mathcal{N}(\bm{0}, \bm{\Sigma}^{-1}_{\mathbf{P}}) \\
	\sqrt{n}\left(  \mathbf{\tp}_{Y^{n}} -\mathbf{P} \right) &\xrightarrow{d}& \mathcal{N}(\bm{0}, \bm{\Sigma}^{-1}_{\mathbf{P}}) 	
	\label{eq:typecd}
\end{eqnarray}
which implies that
\begin{equation}
	\mathbf{V}_{n} \triangleq	\sqrt{n}\left(  \mathbf{\tp}_{X^{n}} -\mathbf{\tp}_{Y^{n}} \right) \xrightarrow{d} \mathcal{N}(\bm{0}, 2\bm{\Sigma}^{-1}_{\mathbf{P}}). \label{eq:(51)}
\end{equation}
Here, $\xrightarrow{d}$ denotes convergence in distribution, and $\mathcal{N}(\mu,\Sigma)$ is used to denote either a Gaussian random vector of mean $\mu$ and covariance matrix $\Sigma$ or its distribution. From \eqref{eq:div_quad}, it follows that
\begin{IEEEeqnarray}{lCl}
	D( \tp_{X^{n}}\| \tp_{Y^{n}})   
	&= & \frac{1}{n} \mathbf{V}_{n}^{\mathsf{T}}  \bm{A}_{D,\mathbf{P}} \mathbf{V}_{n}+o_{P}\left( \frac{1}{n}\right) \label{eq:divstat}
\end{IEEEeqnarray}
where we say that a sequence of random variables $\{X_n\}$ is $X_n = o_p(a_n)$ if for every $\varepsilon>0$
\begin{equation}
\lim_{n\to\infty} P\left(\left|\frac{X_n}{a_n}\right| > \varepsilon\right) = 0.
\end{equation}
It then follows from \eqref{eq:(51)} and \cite[Prop.~6.3.4]{BDB91} that
\begin{IEEEeqnarray}{lCl}
	nD( \tp_{X^{n}}\| \tp_{Y^{n}})   
	&= & \mathbf{V}_{n}^{\mathsf{T}}  \bm{A}_{D,\mathbf{P}} \mathbf{V}_{n} \xrightarrow{d}  \mathbf{V}^{\mathsf{T}}  \bm{A}_{D,\mathbf{P}} \mathbf{V}
\end{IEEEeqnarray}
where $\mathbf{V} \sim \mathcal{N}(\bm{0}, 2\bm{\Sigma}^{-1}_{\mathbf{P}})$. 
It can be shown that \begin{IEEEeqnarray}{lCl}
	\mathbf{V}^{\mathsf{T}}  \bm{A}_{D,\mathbf{P}} \mathbf{V}=2 \sum_{i=1}^{k-1} \lambda_{i} U_{i}^{2}
\end{IEEEeqnarray}
where $U_{1}, \ldots, U_{k-1}$ are i.i.d.  standard normal random variables, and $\lambda_{i}$ are the eigenvalues of $\bm{\Sigma}_{\mathbf{P}}^{-1/2}  \bm{A}_{D,\mathbf{P}} \bm{\Sigma}_{\mathbf{P}}^{-1/2}$. Thus, the test statistic $	\frac{n}{2}D( \tp_{X^{n}}\| \tp_{Y^{n}})   $ converges in distribution to the generalized chi-square distribution  $\chi^{2}_{\bm{\lambda},k-1}$  with vector parameter 
$\bm{\lambda}=(\lambda_{1}, \ldots, \lambda_{k-1})^{\mathsf{T}}$ and $k-1$ degrees of freedom. By \cite[Th.~7.6.2]{RC199}, this convergence is, in fact, uniform. Consequently, for all $c>0$,
\begin{equation}
	P^{n}\left(\frac{n}{2}		D( \tp_{X^{n}}\| \tp_{Y^{n}})  \geq c\right) =  \mathsf{Q}_{\chi^{2}_{\bm{\lambda},k-1}}(c) + O(\delta_{n}) \label{eq:ratediv2d}
\end{equation}
for some positive sequence  $\{\delta_{n}\}$ that is independent of $c$ and vanishes as $n\to\infty$. 

\section{Proof of Theorem~\ref{Thm:unknownPdivtest}}
\label{sec:proof}

From \eqref{eq:ratediv_inv}, it follows that there exist $M_{0}>0$ and $N_{0} \in \mathbb{N}$ such that, for any $P \in \mathcal{P}(\mathcal{Z})$,
\begin{equation}
	\left| P^{n} \left( \frac{n}{2} D(\tp_{X^{n}} \| \tp_{Y^{n}} )\geq  c \right) - \mathsf{Q}_{\chi^{2}_{k-1}}(c/\eta) \right|  \leq M_{0} \delta_{n}  \label{eq:robratekl1a} 
\end{equation}
for $n \geq N_{0}$, and $\{\delta_{n}\}$ as defined in \eqref{eq:ratediv2d1c}. Since the type-II error is monotonically increasing in the threshold, the threshold value $r_n$ that minimizes $\beta_n(\textsf{T}_n^{D}(r_{n}))$ is the smallest value $r_{n}$ for which $\alpha_{n}(\mc_{n}^{D}(r_n)) \leq \epsilon$, which we denote as $r^{D}_{n, \epsilon}$. Mathematically, we define, for $0 <\epsilon <1$ and $n \in \mathbb{N}$,
\begin{subequations}
	\begin{IEEEeqnarray}{rCl}
		\mathcal{R}_{n, \epsilon}^{D}  &\triangleq & \left\lbrace  r > 0 \colon P^{n} \left( D(\tp_{X^n} \|\tp_{Y^{n}}))  \geq r \right)  \leq \epsilon \right\rbrace \label{eq:rnsetrb} \\
		r^{D}_{n, \epsilon} & \triangleq  & \inf \mathcal{R}_{n,\epsilon}^{D}. \label{eq:rnrbdi}
	\end{IEEEeqnarray}
    \end{subequations}
    By definition, if $r_n < r^{D}_{n, \epsilon}$, then the type-I error exceeds $\epsilon$. We can thus assume without loss of optimality that $r_n \geq r^{D}_{n, \epsilon}$. 
    
	It can be shown along the lines of the proof of \cite[Lemma~6]{HJT25} that
	\begin{equation}
		r^{D}_{n, \epsilon} = \frac{2\eta }{n} \mathsf{Q}^{-1}_{\chi^{2}_{k-1}} ( \epsilon)+  O \left( \frac{\delta_{n}}{n}  \right). \label{eq:sec6rbdi}
	\end{equation}
    If there is a threshold value $r_n\in\mathcal{R}_{n, \epsilon}^{D}$ that attains the infimum in \eqref{eq:rnrbdi}, then $\alpha_{n}(\mc_{n}^{D}(r^{D}_{n, \epsilon})) \leq \epsilon$ and it suffices to directly analyze $\beta_n(\textsf{T}_n^{D}(r^D_{n,\epsilon}))$. In general,
    \begin{equation}
        \bar{r}^{D}_{n, \epsilon} \triangleq r^{D}_{n, \epsilon} + \frac{\delta_n}{n}, \quad n\in\mathbb{N} \label{eq:rbrvalue1}
    \end{equation}
    will lie in $\mathcal{R}_{n, \epsilon}^{D}$ and therefore satisfy $\alpha_{n}(\mc_{n}^{D}(\bar{r}^{D}_{n, \epsilon})) \leq \epsilon$. It then follows from the monotonicity of $\beta_n(\textsf{T}_n^{D}(r_{n}))$ in $r_n$ that
    \begin{IEEEeqnarray}{lCl}
        \beta_n(\textsf{T}_n^{D}(r^{D}_{n, \epsilon})) & \leq & \inf_{r_n\colon \alpha_n(\mc_{n}^{D}(r_n))\leq \epsilon}  \beta_n(\textsf{T}_n^{D}(r_n)) \nonumber\\
    & \leq & \beta_n(\textsf{T}_n^{D}(\bar{r}^{D}_{n, \epsilon})). \label{eq:(52)}
    \end{IEEEeqnarray}
    In Subsection~\ref{sub:UB}, we derive an upper bound on $\beta_n(\textsf{T}_n^{D}(\bar{r}^{D}_{n, \epsilon}))$; in Subsections~\ref{sub:LB}, we derive a lower bound on $\beta_n(\textsf{T}_n^{D}(r^{D}_{n, \epsilon}))$. Since both bounds have the same second-order asymptotic behavior, Theorem~\ref{Thm:unknownPdivtest} follows.

\subsection{Upper Bound on $\beta_n$}
\label{sub:UB}
Using the method of types \cite[Th.~11.1.4]{ATCB}, the type-II error for the test $\textsf{T}_n^{D}(\bar{r}^{D}_{n, \epsilon})$ can be bounded as
\begin{IEEEeqnarray}{lCl}
	\IEEEeqnarraymulticol{3}{l}{\beta_n(\textsf{T}_n^{D}(\bar{r}^{D}_{n, \epsilon}))} \nonumber\\
    \quad &= & \sum_{(P',P'') \in 	\mathcal{B}_{D}(\bar{r}^{D}_{n, \epsilon}) \cap (\mathcal{P}_n \times \mathcal{P}_n)} P_1^n(\mathcal{T}(P')) P_2^n(\mathcal{T}(P'')) \notag \\
	&\leq & \sum_{(P',P'') \in 	\mathcal{B}_{D}(\bar{r}^{D}_{n, \epsilon})\cap (\mathcal{P}_n \times \mathcal{P}_n)} e^{-n D_{\textnormal{KL}}(P'\| P_1) - n D_{\textnormal{KL}}(P'' \| P_2)}  \nonumber\\\label{eq:expupp}
\end{IEEEeqnarray}
where $\mathcal{B}_{D}$ was defined in \eqref{eq:B_D(r)} and $\mathcal{T}(\cdot)$ denotes the type class. 

We next lower-bound $D_{\text{KL}}(P' \| P_{1})+D_{\text{KL}}(P'' \| P_{2})$ over $(P', P'') \in \mathcal{B}_{D}(r^{D}_{n, \epsilon})$. To this end, we first note that $\bar{r}^{D}_{n, \epsilon}$ defined in \eqref{eq:rbrvalue1} satisfies $\bar{r}^{D}_{n, \epsilon}=\Theta(1/n)$, so for every $(P',P'') \in \mathcal{B}_{D}(\bar{r}^{D}_{n, \epsilon})$ we have that $D(P'\| P'') = O(1/n)$. it follows from \cite[Lemma~2]{HJT25} that $\|\mathbf{P}'-\mathbf{P}''\|_{2} = O(1/\sqrt{n})$, which implies that $\|P'-P''\|_{1}=O(1/\sqrt{n})$, where $\|\cdot\|_{1}$ is the $\ell_{1}$-norm.  This in turn yields that
\begin{equation}
	\|P'-\bar{P}\|_{1} +\|P''-\bar{P}\|_{1} =O(1/\sqrt{n}) \label{eq:orderl1}
\end{equation}
for $\bar{P} = \frac{1}{2}P' + \frac{1}{2} P''$.

Since $P^{*} \in \mathcal{P}(\mathcal{Z})$, we can find a $\delta>0$ (independent of $n$) such that the $\ell_1$-ball
\begin{IEEEeqnarray}{lCl}
\IEEEeqnarraymulticol{3}{l}{\mathcal{B}_{\ell_1,P^*}(\delta)} \nonumber\\
&\triangleq & \{(T,R)\in\mathcal{P}^2(\mathcal{Z}))\colon  \|T-P^*\|_1 +\|R-P^*\|_1 \leq \delta\} \IEEEeqnarraynumspace\label{eq:BL1}
\end{IEEEeqnarray}
is contained in $\mathcal{P}^2(\mathcal{Z})$. To lower-bound $D_{\text{KL}}(P'\| P_{1})+D_{\text{KL}}(P'' \| P_{2})$ over $(P',P'')\in\mathcal{B}_{D}(\bar{r}^{D}_{n, \epsilon}) $, we distinguish between the cases  $(\bar{P}, \bar{P}) \notin\mathcal{B}_{\ell_1,P^*}(\delta)$ and $(\bar{P}, \bar{P}) \in \mathcal{B}_{\ell_1,P^*}(\delta)$. In the former case, we note that, for every $\xi>0$, we can find a sufficiently large $n_{\xi}$ such that, for $n\geq n_{\xi}$,
\begin{IEEEeqnarray}{lCl}
    \IEEEeqnarraymulticol{3}{l}{\inf_{(P', P'')\in \mathcal{B}_{D}(\bar{r}^{D}_{n, \epsilon})} \{D_{\text{KL}}(P'\| P_{1})+D_{\text{KL}}(P'' \| P_{2})\}} \notag \\
	& \geq & \inf_{(P', P'')\in \mathcal{B}_{D}(\bar{r}^{D}_{n, \epsilon})} \{D_{\text{KL}}(\bar{P}\| P_{1})+D_{\text{KL}}(\bar{P} \| P_{2})\} - \xi \IEEEeqnarraynumspace
\end{IEEEeqnarray}
which follows from \eqref{eq:orderl1} and the continuity of  KL divergence. Recall that $(\bar{P},\bar{P})$ is in $\mathcal{C}$. So if $(\bar{P},\bar{P}) \notin \mathcal{B}_{\ell_1,P^*}(\delta)$, then
\begin{IEEEeqnarray}{lCl}
    \IEEEeqnarraymulticol{3}{l}{\inf_{(P', P'')\in \mathcal{B}_{D}(\bar{r}^{D}_{n, \epsilon})} \{D_{\text{KL}}(P'\| P_{1})+D_{\text{KL}}(P'' \| P_{2})\}} \notag \\
    \quad & >  & D_{\text{KL}}(P^{*}\| P_{1})+D_{\text{KL}}(P^{*} \| P_{2}) \label{eq:(66)}
\end{IEEEeqnarray}
since $P^*$ is the unique minimizer of $ D_{\textnormal{KL}}(P  \|  P_{1})+D_{\textnormal{KL}}(P  \|  P_{2})$. We can thus find a sufficiently small $\xi>0$ such that
since $P^*$ is the unique minimizer of $ D_{\textnormal{KL}}(P  \|  P_{1})+D_{\textnormal{KL}}(P  \|  P_{2})$. We can thus find a sufficiently small $\xi>0$ such that
\begin{IEEEeqnarray}{lCl}
	\IEEEeqnarraymulticol{3}{l}{\inf_{(P', P'')\in \mathcal{B}_{D}(\bar{r}^{D}_{n, \epsilon})} \{D_{\text{KL}}(P'\| P_{1})+D_{\text{KL}}(P'' \| P_{2})\}}\nonumber\\
	\quad & \geq & D_{\text{KL}}(P^{*}\| P_{1})+D_{\text{KL}}(P^{*} \| P_{2}) \notag \\
	&& {} - \sqrt{\frac{V_{\text{KL}}(P^{*}\|P_1) + V_{\text{KL}}(P^{*}\|P_2) }{n}} \sqrt{\mathsf{Q}^{-1}_{\chi^{2}_{k-1}}(\epsilon)}. \IEEEeqnarraynumspace\label{eq:ROB_P'_neq_P*}
\end{IEEEeqnarray}

To treat the case where  $(\bar{P},\bar{P}) \in \mathcal{B}_{\ell_1,P^*}(\delta)$, we perform a Taylor-series approximation of $D_{\text{KL}}(P'\| P_{1})+D_{\text{KL}}(P'' \| P_{2})$ around $(\bar{P},\bar{P})$ and minimize the approximation over $(P',P'') \in \mathcal{B}_{D}(\bar{r}^{D}_{n, \epsilon})$. Indeed, following along the lines of the proof of \cite[Lemma~4]{HJT25}, it can be shown that 
        there exist $M_{1}>0$, and $N_1$ such that
		\begin{IEEEeqnarray}{lCl}
			\IEEEeqnarraymulticol{3}{l}{D_{\text{KL}}(P' \| P_1) + D_{\text{KL}}(P''\| P_2) }\notag \\
			\quad & \geq & D_{\text{KL}}(\bar{P}\| P_1) + D_{\text{KL}}(\bar{P}\| P_2) +\ell_{\bar{P}}(P',P'')  - \frac{M_{1}}{n^{3/2}} \IEEEeqnarraynumspace\label{eq:expo_lb_p1p2}
		\end{IEEEeqnarray}
		for $n \geq N_1$ and $(P',P'')\in \mathcal{B}_{D}(\bar{r}^{D}_{n, \epsilon})$, where
		\begin{eqnarray}
			\ell_{\bar{P}}(T,R) \triangleq 
			\sum_{i=1}^k (T_i - \bar{P}_{i}) \ln \frac{\bar{P}_{i}}{P_{1i}}  + \sum_{i=1}^k (R_i - \bar{P}_{i}) \ln \frac{\bar{P}_{i}}{P_{2i}} \label{eq:minfun1a}
		\end{eqnarray}
        for $T,R, \bar{P} \in \mathcal{P}(\mathcal{Z})$. To minimize \eqref{eq:expo_lb_p1p2} over $(P',P'') \in \mathcal{B}_{D}(\bar{r}^{D}_{n, \epsilon})$, we need the following lemma.
		\begin{lemma}\label{Thm:div_ball}
			 There exist constants $M'>0$ and  $N' \in \mathbb{N}$ such that, for all $n \geq N'$, the pair $(P',P'')\in\mathcal{B}_{D}(\bar{r}^{D}_{n, \epsilon})$ lies in the set $\mathcal{A}_{\bm{\Sigma}_{\bar{\mathbf{P}}}}(r_{n}')$, where $\mathcal{A}_{\bm{\Sigma}_{\bar{\mathbf{P}}}}$ was defined in \eqref{eq:A_Sigma} and
			\begin{equation}
				r_{n}'\triangleq \frac{\bar{r}^{D}_{n, \epsilon}}{2\eta}+\frac{M'}{n^{3/2}}. \label{eq:rnbar}
			\end{equation}
		\end{lemma}
		\begin{IEEEproof}
			See Appendix~\ref{sec:append_div_ball}. 
		\end{IEEEproof}
        It follows from Lemma~\ref{Thm:div_ball} that  we can lower-bound $\ell_{\bar{P}}(P',P'')$ for $(P',P'')\in\mathcal{B}_{D}(\bar{r}^{D}_{n, \epsilon})$ as
        \begin{equation}
            \ell_{\bar{P}}(P',P'') \geq \min_{(T,R)\in \mathcal{A}_{\bm{\Sigma}_{\bar{\mathbf{P}}}}(r_{n}')} \ell_{\bar{P}}(T,R) \triangleq \ell^{*}_{\bar{P}}(r_{n}').
        \end{equation}
       This minimum is evaluated in the following lemma:
		\begin{lemma}\label{Thm:min_quad}
			For $(P',P'') \in \mathcal{B}_{D}(r_{n}')$, $\bar{P}=\frac{1}{2}P' + \frac{1}{2}P''$, and $r_{n}'$ defined in \eqref{eq:rnbar}, the minimum value $\ell^{*}_{\bar{P}}(r_{n}')$ is given by
			\begin{equation}
				\ell^{*}_{\bar{P}}(r_{n}') =
				-\sqrt{r_{n}'}\sqrt{	V_{\text{KL}}(\bar{P}\|P_1) + V_{\text{KL}}(\bar{P}\|P_2)}.
				\label{eq:barrvalue}
			\end{equation} 
		\end{lemma}
		\begin{IEEEproof}
			See Appendix~\ref{sec:appendA_min_quad}.
		\end{IEEEproof}
		
        Applying Lemma~\ref{Thm:min_quad} to \eqref{eq:expo_lb_p1p2}, we thus obtain that
		\begin{IEEEeqnarray}{lCl}
		\IEEEeqnarraymulticol{3}{l}{D_{\text{KL}}(P' \| P_1) + D_{\text{KL}}(P''\| P_2)}\nonumber\\
			\quad	& \geq &    D_{\text{KL}}(\bar{P}\| P_1) + D_{\text{KL}}(\bar{P}\| P_2)   \notag \\
			& & {} -\sqrt{r_{n}'}\sqrt{	V_{\text{KL}}(\bar{P}\|P_1) + V_{\text{KL}}(\bar{P}\|P_2)} -  \frac{ M_{2} }{n^{3/2}}. \IEEEeqnarraynumspace\label{eq:lowerb}
	\end{IEEEeqnarray}
    As shown in Appendix~\ref{sec:append_min_seq}, the minimizing distribution 
		\begin{IEEEeqnarray}{lCl}
			\tilde{P}_{n} & \triangleq  & \arg \min_{P \in \bar{\mathcal{P}}(\mathcal{Z})} \Bigl\{ D_{\text{KL}}(P\| P_1) + D_{\text{KL}}(P\| P_2) \notag \\
			& & {}  -\sqrt{r_{n}'}\sqrt{	V_{\text{KL}}(P\|P_1) + V_{\text{KL}}(P\|P_2)}\Bigr\} 
			\label{eq:tylq4e}
		\end{IEEEeqnarray}
		exists and converges to $P^*$ as $n\to\infty$. Since $r_n'=O(1/n)$, it thus follows that
	\begin{IEEEeqnarray}{lCl}
		\IEEEeqnarraymulticol{3}{l}{\sqrt{r_{n}'}\sqrt{	V_{\text{KL}}(\tilde{P}_{n}\|P_1) + V_{\text{KL}}(\tilde{P}_{n}\|P_2)}} \notag \\
		\quad &= & \sqrt{r_{n}'}\sqrt{	V_{\text{KL}}(P^{*}\|P_1) + V_{\text{KL}}(P^{*}\|P_2)}+ o\left(\frac{1}{\sqrt{n}}\right). \IEEEeqnarraynumspace \label{eq:2nd_P*}
	\end{IEEEeqnarray}
    Minimizing \eqref{eq:lowerb} over $\bar{P}$, we then obtain that
    \begin{IEEEeqnarray}{lCl}
\IEEEeqnarraymulticol{3}{l}{D_{\text{KL}}(P' \| P_1) + D_{\text{KL}}(P''\| P_2)}\nonumber\\
& \geq &    D_{\text{KL}}(\tilde{P}_n\| P_1) + D_{\text{KL}}(\tilde{P}_n\| P_2)   \notag \\
			& & {} -\sqrt{r_{n}'}\sqrt{	V_{\text{KL}}(\tilde{P}_n\|P_1) + V_{\text{KL}}(\tilde{P}_n\|P_2)} -  \frac{ M_{2} }{n^{3/2}} \nonumber\\
            & \geq & D_{\text{KL}}(P^{*}\| P_1) + D_{\text{KL}}(P^{*}\| P_2)  \nonumber\\
            & & {} - \sqrt{r_{n}'}\sqrt{	V_{\text{KL}}(P^{*}\|P_1) + V_{\text{KL}}(P^{*}\|P_2)} + o\left(\frac{1}{\sqrt{n}}\right) \nonumber\\
            & \geq & D_{\text{KL}}(P^{*}\| P_1) + D_{\text{KL}}(P^{*}\| P_2)  \nonumber\\
            & & {} - \sqrt{	\frac{V_{\text{KL}}(P^{*}\|P_1) + V_{\text{KL}}(P^{*}\|P_2)}{n}} \sqrt{\mathsf{Q}^{-1}_{\chi^{2}_{k-1}}(\epsilon)} + o\left(\frac{1}{\sqrt{n}}\right) \nonumber\\\label{eq:tylq4cc}
    \end{IEEEeqnarray}
    where the second inequality follows because $P^*$ minimizes $D_{\text{KL}}(\tilde{P}_n\| P_1) + D_{\text{KL}}(\tilde{P}_n\| P_2)$ over $\tilde{P}_n$, and by applying \eqref{eq:2nd_P*}; the subsequent inequality follows by performing a Taylor-series expansion of $\sqrt{r_n'}$ around $\frac{1}{\sqrt{n}}\sqrt{\mathsf{Q}^{-1}_{\chi^{2}_{k-1}}( \epsilon)}$, which yields
    \begin{equation}
        \sqrt{r_{n}' }\leq \frac{1}{\sqrt{n}}\sqrt{\mathsf{Q}^{-1}_{\chi^{2}_{k-1}}( \epsilon)} + o\left(\frac{1}{\sqrt{n}}\right).
        \end{equation}

  Since $	D_{\text{KL}}(P^{*}\| P_1) + D_{\text{KL}}(P^{*}\| P_2)	=2D_{B}(P_{1},P_{2})$, we obtain from \eqref{eq:(52)}, \eqref{eq:expupp}, \eqref{eq:ROB_P'_neq_P*}, and \eqref{eq:tylq4cc}, and the fact that the number of types $|\mathcal{P}_n|$ is bounded by $(n+1)^{|\mathcal{Z}|}$ \cite[Th.~11.1.1]{ATCB}, that
	\begin{IEEEeqnarray}{lCl}
		\sup_{r_n\colon \alpha_n(\mc_{n}^{D}(r_n))\leq \epsilon} -\frac{\ln	\beta_n(\mc_{n}^{D}(r_n))}{n} & \geq &   2 D_{B}(P_{1},P_{2}) \nonumber\\
		 \IEEEeqnarraymulticol{3}{r}{{}-\sqrt{\frac{V_{\text{KL}}(P^{*}\|P_1) + V_{\text{KL}}(P^{*}\|P_2) }{n}} \sqrt{\mathsf{Q}^{-1}_{\chi^{2}_{k-1}}(\epsilon)}+o\biggl(\frac{1}{\sqrt{n}}\biggr)}. \nonumber\\
	\end{IEEEeqnarray}
    
	\subsection{Lower Bound on $\beta_n$}
    \label{sub:LB}
To lower-bound $\beta_n(\textsf{T}_n^{D}(r^{D}_{n, \epsilon}))$, we first restrict the set $\mathcal{B}_{D}(r^{D}_{n, \epsilon})$ as in the following lemma:
\begin{lemma} \label{Thm:div_ball_converse} There exist constants $M_3>0$ and $N_3\in\mathbb{N}$ such that, for all $n\geq N_3$,
\begin{equation}
    \mathcal{A}_{\bm{\Sigma}_{\mathbf{P}^*}}(\tilde{r}_{n}) \subseteq \mathcal{B}_{D}(r_{n,\epsilon}^{D}) \label{eq:klchiwrb1}
\end{equation}
where
\begin{equation}
\tilde{r}_n \triangleq \frac{r^{D}_{n, \epsilon}}{2\eta}-\frac{M_{3}}{ n^{3/2} }. \label{eq:tilde_rn}
\end{equation}
	\end{lemma}
	\begin{IEEEproof}
		See Appendix~\ref{sec:append_div_ball_converse}. 
	\end{IEEEproof}
It then follows from the method of types \cite[Th.~11.1.4]{ATCB} that
\begin{IEEEeqnarray}{lCl}
\IEEEeqnarraymulticol{3}{l}{\beta_n(\textsf{T}_n^{D}(r^{D}_{n, \epsilon}))} \nonumber\\
\quad & = & \sum_{(P',P'') \in 	\mathcal{B}_{D}(r^{D}_{n, \epsilon}) \cap (\mathcal{P}_n \times \mathcal{P}_n)} P_1^n(\mathcal{T}(P')) P_2^n(\mathcal{T}(P'')) \nonumber\\
& \geq & \frac{1}{(n+1)^{2|\mathcal{Z}|}} e^{-n D_{\textnormal{KL}}(\tp'_{n} \| P_{1}) - n D_{\textnormal{KL}}(\tp''_{n} \| P_{2})}  \label{eq:prtyclcrb1} 
\end{IEEEeqnarray}
for some type distributions $\tp'_{n}$ and $\tp''_{n}$ that lie in $\mathcal{A}_{\bm{\Sigma}_{\mathbf{P}^*}}(\tilde{r}_{n})$ and satisfy
\begin{equation}
    |n\ell_{P^{*}}(\Gamma',\Gamma'')- n \ell_{P^{*}}(\tp'_{n}, \tp''_{n})| \leq \kappa
\end{equation}
for some constant $\kappa > 0$, where $(\Gamma',\Gamma'')$ is the pair of distributions that minimizes $\ell_{P^{*}}(P',P'')$ over all pairs $(P',P'')$ in $\mathcal{A}_{\bm{\Sigma}_{\mathbf{P}^*}}(\tilde{r}_{n})$. (Recall that $\mathcal{A}_{\bm{\Sigma}_{\mathbf{P}^*}}$ is as in \eqref{eq:A_Sigma} but with $\bar{\mathbf{P}}$ replaced by $\mathbf{P}^*$.) The existence of such a pair $(\tp'_{n}, \tp''_{n}) $ is established in the following lemma:

	\begin{lemma}\label{Thm:type}  Let $(\Gamma',\Gamma'')$ be the pair of distributions that minimizes $\ell_{P^{*}}(P',P'')$  over all pairs $(P',P'')$ in $\mathcal{A}_{\bm{\Sigma}_{\mathbf{P}^*}}(\tilde{r}_{n})$. Then, there exists a constant $\tilde{N} \geq N_3$ and a pair of type distributions $(\tp'_{n},\tp''_{n})$ in $\mathcal{A}_{\bm{\Sigma}_{\mathbf{P}^*}}(\tilde{r}_{n})$ satisfying
		\begin{equation}
			|n \ell_{P^{*}}(\Gamma',\Gamma'') -n\ell_{P^{*}}(\tp'_{n},\tp''_{n})| \leq \kappa  \label{eq:hfunbound1}
		\end{equation}
		for $n\geq \tilde{N}$ and some constant $\kappa >0$.
	\end{lemma}
	\begin{IEEEproof} See Appendix~\ref{sec:appendA_lemmatype}.
	\end{IEEEproof}
    
	Since $(\tp'_{n},\tp''_{n})\in\mathcal{A}_{\bm{\Sigma}_{\mathbf{P}^*}}(\tilde{r}_{n})$, we have that $\|\tp'_{n}-P^{*} \|_{2}+\|\tp''-P^{*} \| =O(1/\sqrt{n})$. Then, by following the similar steps as in \cite[Lemma~2]{HJT25}, we obtain that there are constants $M_4>0$ and $N_4 \in \mathbb{N}$ such that, for $n \geq N_4$, 
	\begin{IEEEeqnarray}{lCl}
		\IEEEeqnarraymulticol{3}{l}{D_{\textnormal{KL}}(\tp'_{n} \| P_{1})+D_{\textnormal{KL}}(\tp''_{n} \| P_{2})}\nonumber\\
		\quad & \leq & D_{\textnormal{KL}}(P^{*} \| P_{1}) +D_{\textnormal{KL}}(P^{*} \| P_{2}) +\ell_{P^{*}}(P'_{n}, P''_{n})\nonumber \\
		& & {} +\frac{1}{2} d_{\chi^{2}} (\tp'_{n}, P^{*}) +\frac{1}{2} d_{\chi^{2}} (\tp''_{n}, P^{*}) + \frac{M'_{2}}{n^{3/2}} \nonumber\\
		& \leq &  D_{\textnormal{KL}}(P^{*} \| P_{1}) +D_{\textnormal{KL}}(P^{*} \| P_{2}) \nonumber \\
		& & {} - \sqrt{\tilde{r}_{n}} \sqrt{V_{\text{KL}}(P^{*}\|P_1) + V_{\text{KL}}(P^{*}\|P_2)} +\frac{\kappa}{n} \nonumber\\
		& & {} + \frac{1}{2} d_{\chi^{2}} (\tp'_{n}, P^{*}) +\frac{1}{2} d_{\chi^{2}} (\tp'', P^{*}) + \frac{M'_{2}}{n^{3/2}} \label{eq:KL_uppr1rb}
	\end{IEEEeqnarray}
	where the last inequality follows from \eqref{eq:hfunbound1} in Lemma~\ref{Thm:type}. Performing a Taylor-series expansion of  $\sqrt{\tilde{r}_n}$ around $\frac{1}{\sqrt{n}}\sqrt{\mathsf{Q}^{-1}_{\chi^{2}_{k-1}}( \epsilon)}$, we obtain that
	\begin{equation}
    \sqrt{\tilde{r}_{n} }= \frac{1}{\sqrt{n}}\sqrt{\mathsf{Q}^{-1}_{\chi^{2}_{k-1}}( \epsilon)} + o\left(\frac{1}{\sqrt{n}}\right). \label{eq:(75)}
    \end{equation}
Since $	D_{\text{KL}}(P^{*}\| P_1) + D_{\text{KL}}(P^{*}\| P_2)	=2D_{B}(P_{1},P_{2})$, we then obtain from \eqref{eq:(52)}, \eqref{eq:prtyclcrb1}, \eqref{eq:KL_uppr1rb}, and \eqref{eq:(75)} that
	\begin{IEEEeqnarray}{lCl}
		\sup_{r_n\colon \alpha_n(\mc_{n}^{D}(r_n))\leq \epsilon} -\frac{\ln	\beta_n(\mc_{n}^{D}(r_n))}{n} & \leq &   2 D_{B}(P_{1},P_{2}) \nonumber\\
		 \IEEEeqnarraymulticol{3}{r}{{}-\sqrt{\frac{V_{\text{KL}}(P^{*}\|P_1) + V_{\text{KL}}(P^{*}\|P_2) }{n}} \sqrt{\mathsf{Q}^{-1}_{\chi^{2}_{k-1}}(\epsilon)}+o\biggl(\frac{1}{\sqrt{n}}\biggr)}. \nonumber\\
	\end{IEEEeqnarray}
    
		\section{Proof of Lemma~\ref{Thm:div_ball}}
	\label{sec:append_div_ball}
    Applying the Taylor-series approximation \eqref{eq:div_quad} with $\mathbf{T}=\mathbf{P}'$, $\mathbf{R}=\mathbf{P}''$, $\mathbf{P}=\bar{\mathbf{P}}$, and $\bm{A}_{D, \mathbf{P}}=\eta \bm{\Sigma}_{\bar{\mathbf{P}}}$, we can approximate $D(P'\|P'')$ as
	\begin{IEEEeqnarray}{lCl}
		D(P' \| P'') & = & \eta (\mathbf{P'}-\mathbf{P''})^{T}\bm{\Sigma}_{\mathbf{\bar{P}}}(\mathbf{P'}-\mathbf{P''}) \notag \\ 
		& & {} +O(\|\mathbf{P'}-\mathbf{\bar{P}}\|_{2}^{3})+O(\|\mathbf{P''}-\mathbf{\bar{P}}\|_{2}^{3}).\IEEEeqnarraynumspace \label{eq:div_quad3}
	\end{IEEEeqnarray}
    Furthermore, for $\bar{\mathbf{P}} = \frac{1}{2}\mathbf{P}'+\frac{1}{2}\mathbf{P}''$, we have that
	\begin{IEEEeqnarray}{lCl}
		\IEEEeqnarraymulticol{3}{l}{\eta (\mathbf{P'}-\mathbf{P''})^{T}\bm{\Sigma}_{\mathbf{\bar{P}}}(\mathbf{P'}-\mathbf{P''})} \notag \\ 
		\quad &  = & 2\eta (\mathbf{P'}-\mathbf{\bar{P}})^{T}\bm{\Sigma}_{\mathbf{\bar{P}}}(\mathbf{P'}-\mathbf{\bar{P}}) \nonumber\\
        & & {} +2\eta (\mathbf{P''}-\mathbf{\bar{P}})^{T}\bm{\Sigma}_{\mathbf{\bar{P}}}(\mathbf{P''}-\mathbf{\bar{P}}).\IEEEeqnarraynumspace\label{eq:div_quad6}
	\end{IEEEeqnarray}

    We next note that, for all $(P',P'') \in \mathcal{B}_{D}(\bar{r}^{D}_{n, \epsilon})$ with $\bar{r}^{D}_{n, \epsilon}$ defined in \eqref{eq:rbrvalue1}, we have $D(P',P'') < \bar{r}^{D}_{n, \epsilon}$.  Since $\bar{r}^{D}_{n, \epsilon}=\Theta(1/n)$, it follows from \cite[Lemma~2]{HJT25} that $\|P'-P''\|_{1}=O(1/\sqrt{n})$. We further have that  $P'-\bar{P}=(P'-P'')/2$ and $P''-\bar{P}=(P''-P')/2$. This implies that  $\|\mathbf{P'}-\bar{\mathbf{P}}\|_{2}$ and $\|\mathbf{P''}-\mathbf{\bar{P}}\|_{2}$ are both of order $1/\sqrt{n}$, too. Applying this to \eqref{eq:div_quad3}, we obtain that there exist constants $\bar{M}>0$ and $\bar{N}\in\mathbb{N}$ such that, for $n\geq \bar{N}$,
    \begin{equation}
		|D(P' \| P'') - \eta (\mathbf{P'}-\mathbf{P''})^{T}\bm{\Sigma}_{\mathbf{\bar{P}}}(\mathbf{P'}-\mathbf{P''}) | < \frac{ \bar{M}}{n^{3/2}} \label{eq:tylp2w}
	\end{equation}
    for all $(P',P'') \in \mathcal{B}_{D}(\bar{r}^{D}_{n, \epsilon})$.
	Together with \eqref{eq:div_quad6}, this implies that if $D(P',P'') <\bar{r}^{D}_{n, \epsilon}$, then
    \begin{IEEEeqnarray}{lCl}
\IEEEeqnarraymulticol{3}{l}{(\mathbf{P'}-\mathbf{\bar{P}})^{T}\bm{\Sigma}_{\mathbf{\bar{P}}}(\mathbf{P'}-\mathbf{\bar{P}})+(\mathbf{P''}-\mathbf{\bar{P}})^{T}\bm{\Sigma}_{\mathbf{\bar{P}}}(\mathbf{P''}-\mathbf{\bar{P}})} \nonumber\\
\quad & \leq & \frac{1}{2\eta} D(P'\|P'') + \frac{\bar{M}}{2\eta n^{3/2}} \nonumber\\
& \leq & \frac{\bar{r}^{D}_{n, \epsilon}}{2\eta} + \frac{\bar{M}}{2\eta n^{3/2}}.
    \end{IEEEeqnarray}
    Thus, setting $r_n' = \bar{r}^{D}_{n, \epsilon}/(2\eta)$ and $M' = \bar{M}/(2\eta)$, we obtain that any pair of distributions $(P',P'')$ in $\mathcal{B}_{D}(\bar{r}^{D}_{n, \epsilon})$ lies in the set $\mathcal{A}_{\bm{\Sigma}_{\mathbf{\bar{P}}}}(\bar{r}_{n}')$, which is Lemma~\ref{Thm:div_ball}.
	
	\section{Proof of Lemma~\ref{Thm:min_quad}}
	\label{sec:appendA_min_quad}
    Consider the function $\ell_{P}(T,R)$ defined in \eqref{eq:minfun1a} as
    \begin{IEEEeqnarray}{lCl}
			\ell_{P}(T,R) & \triangleq &
			\sum_{i=1}^k (T_i - P_{i}) \ln \frac{P_{i}}{P_{1i}}  + \sum_{i=1}^k (R_i - P_{i}) \ln \frac{P_{i}}{P_{2i}} \IEEEeqnarraynumspace\label{eq:minfun1ab}
		\end{IEEEeqnarray}
        for $T,R,P \in \mathcal{P}(\mathcal{Z})$. Minimizing $\ell_{P}(T,R) $ over the set $\mathcal{A}_{\Sigma_{P}}(r)$ (defined in \eqref{eq:A_Sigma}) for some $r>0$ is equivalent to solving the following minimization problem: 
	\begin{equation}
		\text{Minimize } \ \tilde{\ell}(\mathbf{x} , \mathbf{y} )=\mathbf{c}^T \mathbf{x} + \mathbf{d}^T \mathbf{y}
	\end{equation}
	subject to the constraint $\mathbf{x}^T \bm{\Sigma}_{P} \mathbf{x} + \mathbf{y}^T \bm{\Sigma}_{P} \mathbf{y} \leq r$, where we denote
    \begin{subequations}
	\begin{IEEEeqnarray}{lCl}
    \mathbf{x} &=& (x_1, \dots, x_{k-1}), \quad  x_i = (T_i - P_i),  \\
		\mathbf{y} & =& (y_1, \dots, y_{k-1}), \quad y_i = (R_i - P_i) \\
		\mathbf{c} &= & (c_1, \dots, c_{k-1}), \quad 
		c_i = \ln \frac{P_i}{P_{1i}} - \ln \frac{P_k}{P_{1k}} \label{eq:vector_c} \\
		\mathbf{d} & = & (d_1, \dots, d_{k-1}), \quad 
		d_i = \ln \frac{P_i}{P_{2i}} - \ln \frac{P_k}{P_{2k}}. \label{eq:vector_d}\\
	\end{IEEEeqnarray}
    \end{subequations}
    Denoting $g_0(\mathbf{x}, \mathbf{y} )=\mathbf{x}^T \bm{\Sigma}_{P} \mathbf{x} + \mathbf{y}^T \bm{\Sigma}_{P} \mathbf{y} -r$, let us consider the Lagrangian function 
	\begin{equation}
		\mathcal{L}(\mathbf{x}, \mathbf{y},  \mu_{0}) = \tilde{\ell}(\mathbf{x}, \mathbf{y} ) + \mu_{0} g_0(\mathbf{x}, \mathbf{y} )
	\end{equation}
	where  $\mu_{0}$ is a KKT multiplier. Evaluating the KKT conditions, we obtain that
	\begin{subequations}
    \begin{IEEEeqnarray}{rCl}
		c_{i}+ \mu_{0} \left(  2 \sum_{j=1}^{k-1} \bm{\Sigma}_{\mathbf{P}}(i,j) x_{j} \right) & = &  0, \quad \forall i \label{eq:kkt1a}\\
		d_{i}+ \mu_{0} \left(  2 \sum_{j=1}^{k-1} \bm{\Sigma}_{\mathbf{P}}(i,j) y_{j} \right) & = & 0, \quad \forall i \label{eq:kkt1ay}\\
		\IEEEeqnarraymulticol{3}{l}{\sum_{i=1}^{k-1} \sum_{j=1}^{k-1} \bm{\Sigma}_{\mathbf{P}}(i,j) x_{i}x_{j}} \nonumber\\
        {} +	\sum_{i=1}^{k-1} \sum_{j=1}^{k-1} \bm{\Sigma}_{\mathbf{P}}(i,j) y_{i}y_{j}- r & \leq & 0 \label{eq:kkt2a} \\
		\mu_{0}  &\geq &  0 \label{eq:kkt5a}  \\
		\IEEEeqnarraymulticol{3}{l}{\mu_{0} \Biggl( 	\sum_{i=1}^{k-1} \sum_{j=1}^{k-1} \bm{\Sigma}_{\mathbf{P}}(i,j) x_{i}x_{j}} \nonumber\\
        \qquad {} +	\sum_{i=1}^{k-1} \sum_{j=1}^{k-1} \bm{\Sigma}_{\mathbf{P}}(i,j) y_{i}y_{j}- r\Biggr)   & = & 0. \label{eq:kkt8a}
	\end{IEEEeqnarray}
    \end{subequations}
	Since, by assumption, $P_1\neq P_2$, the vectors $\mathbf{c}$ and $\mathbf{d}$ are not the all-zero vectors. Consequently, \eqref{eq:kkt1a} and \eqref{eq:kkt1ay} cannot be satisfied for $\mu_0=0$. For $\mu_{0}>0$, we obtain from \eqref{eq:kkt1a} and  \eqref{eq:kkt1ay} that 
	\begin{subequations}
    \begin{IEEEeqnarray}{lCl}
		\mathbf{x} & = & \frac{-1}{2 \mu_{0}}  \bm{\Sigma}_{\mathbf{P}}^{-1} \mathbf{c} \label{eq:kkts4} \\
		\mathbf{y} & = & \frac{-1}{2 \mu_{0}}  \bm{\Sigma}_{\mathbf{P}}^{-1} \mathbf{d} \label{eq:kkts4y}
	\end{IEEEeqnarray}
    \end{subequations}
	where $\bm{\Sigma}_{\mathbf{P}}^{-1}$ exists since $\bm{\Sigma}_{\mathbf{P}} \succ 0$. Furthermore, \eqref{eq:kkt8a} implies that
	\begin{equation}
		\sum_{i=1}^{k-1} \sum_{j=1}^{k-1} \bm{\Sigma}_{\mathbf{P}}(i,j) x_{i}x_{j} +\sum_{i=1}^{k-1} \sum_{j=1}^{k-1} \bm{\Sigma}_{\mathbf{P}}(i,j) y_{i}y_{j}  =r. \label{eq:kkts5}
	\end{equation}
	Substituting \eqref{eq:kkts4} and \eqref{eq:kkts4y} in \eqref{eq:kkts5}, we get
	\begin{equation}
		\frac{1}{4 \mu_{0}^{2}} \left( \mathbf{c}^{\mathsf{T}} \bm{\Sigma}_{\mathbf{P}}^{-1} \mathbf{c} +\mathbf{d}^{\mathsf{T}} \bm{\Sigma}_{\mathbf{P}}^{-1} \mathbf{d} \right) =r\label{eq:kkts7}
	\end{equation}
	which implies that
	\begin{equation}
		\mu_{0}= \frac{1}{2 \sqrt{r}}  \sqrt{\mathbf{c}^{\mathsf{T}} \bm{\Sigma}_{\mathbf{P}}^{-1} \mathbf{c}+\mathbf{d}^{\mathsf{T}} \bm{\Sigma}_{\mathbf{P}}^{-1} \mathbf{d}}. \label{eq:kkts7a}
	\end{equation}
	The above equation combined with \eqref{eq:kkts4} and \eqref{eq:kkts4y} yields the optimal solution
	\begin{subequations}
    \begin{IEEEeqnarray}{lCl}
		\mathbf{x}^{*} &= &\frac{-\sqrt{r} \bm{\Sigma}_{\mathbf{P}}^{-1} \mathbf{c}}{ \sqrt{\mathbf{c}^{\mathsf{T}} \bm{\Sigma}_{\mathbf{P}}^{-1} \mathbf{c}+\mathbf{d}^{\mathsf{T}} \bm{\Sigma}_{\mathbf{P}}^{-1} \mathbf{d}}} \\
		\mathbf{y}^{*} &= & \frac{-\sqrt{r} \bm{\Sigma}_{\mathbf{P}}^{-1} \mathbf{d}}{\sqrt{\mathbf{c}^{\mathsf{T}} \bm{\Sigma}_{\mathbf{P}}^{-1} \mathbf{c}+\mathbf{d}^{\mathsf{T}} \bm{\Sigma}_{\mathbf{P}}^{-1} \mathbf{d}}}.    \label{eq:kkts8}
	\end{IEEEeqnarray}
    \end{subequations}
	Thus, the  minimum value $\ell^{*}_{P}(r) \triangleq \min_{\mathcal{A}_{\Sigma_{P}}(r)} \ell_{P}(T,R)$ is given by
	\begin{IEEEeqnarray}{lCl}
		\ell^{*}_{P}(r)
		&=& -\sqrt{r} \, \sqrt{\mathbf{c}^T \bm{\Sigma}_{P}^{-1} \mathbf{c} + \mathbf{d}^T \bm{\Sigma}_P^{-1} \mathbf{d}} \nonumber\\
		&=&-\sqrt{r} \, \sqrt{V_{\textnormal{KL}}(P\|P_{1})+V_{\textnormal{KL}}(P\|P_{2})} 
	\end{IEEEeqnarray}
	where the last equation follows from \cite[Eq.~(122)]{HJT25}, since $\mathbf{c}^T \bm{\Sigma}_{P}^{-1} \mathbf{c}=V_{\textnormal{KL}}(P\|P_{1})$, and  $\mathbf{d}^T \bm{\Sigma}_{P}^{-1} \mathbf{d}=V_{\textnormal{KL}}(P\|P_{2})$. 
	Thus, by taking $P=\bar{P}$, and  $r=r'_{n}$, we obtain \eqref{eq:barrvalue}.
	
	\section{Convergence of Minimizing Distribution $\tilde{P}_n$}
	\label{sec:append_min_seq}
	Note that the minimum in \eqref{eq:tylq4e} exists for every $n$ because $D_{\textnormal{KL}}$ and $V_{\textnormal{KL}}$ are continuous functions and $\bar{\mathcal{P}}(\mathcal{Z})$ is a compact set. We next prove by contradiction that $\tilde{P}_{n}\to P^*$ as $n\to\infty$. Indeed, since $\bar{\mathcal{P}}(\mathcal{Z})$ is a compact set, the sequence $\{\tilde{P}_{n}\}$ has a converging subsequence $\{\tilde{P}_{n_{k}}\}$ by the Bolzano-Weierstrass theorem. Suppose the limit of this subsequence is $\check{P} \neq P^*$. Since $D_{\textnormal{KL}}$ is continuous, $V_{\textnormal{KL}}$ is bounded, and $r_{n}'$ vanishes as $n\to\infty$, it follows that
	\begin{IEEEeqnarray}{lCl}
		\IEEEeqnarraymulticol{3}{l}{\lim_{k\to\infty} \Bigl\{ D_{\text{KL}}(\tilde{P}_{n_{k}}\| P_1) + D_{\text{KL}}(\tilde{P}_{n_{k}}\| P_2)}  \nonumber \\
		\IEEEeqnarraymulticol{3}{l}{\quad\qquad {} -\sqrt{r_{n_{k}}'} \sqrt{V_{\textnormal{KL}}(\tilde{P}_{n_{k}}\|P_{1})+V_{\textnormal{KL}}(\tilde{P}_{n_{k}}\|P_{2})} \Bigr\} }\nonumber\\
		\quad &  = & D_{\text{KL}}(\check{P} \| P_1) + D_{\text{KL}}(\check{P} \| P_2) \nonumber \\
		& > & D_{\text{KL}}(P^{*}\| P_1) + D_{\text{KL}}(P^{*}\| P_2)   \label{eq:ROB_contra1a}
	\end{IEEEeqnarray}
	where the inequality follows because $P^*$ is the unique minimizer of $D_{\text{KL}}(P\| P_1) + D_{\text{KL}}(P\| P_2)  $.  However, since $P_n$ is the minimizer of \eqref{eq:tylq4e} over $\bar{\mathcal{P}}(\mathcal{Z})$ and $P^* \in\bar{\mathcal{P}}(\mathcal{Z})$, we also have that
	\begin{IEEEeqnarray}{lCl}
	\IEEEeqnarraymulticol{3}{l}{ D_{\text{KL}}(\tilde{P}_{n} \| P_1) + D_{\text{KL}}(\tilde{P}_{n} \| P_2)} \nonumber \\
		\IEEEeqnarraymulticol{3}{l}{\qquad\quad {} -\sqrt{r_{n}'}\sqrt{V_{\textnormal{KL}}(\tilde{P}_{n}\|P_{1})+V_{\textnormal{KL}}(\tilde{P}_{n}\|P_{2})}} \nonumber \\
        \quad & \leq & D_{\text{KL}}(P^* \| P_1) + D_{\text{KL}}(P^* \| P_2) \nonumber \\
        & & {} - \sqrt{r_{n}'}\sqrt{V_{\textnormal{KL}}(P^*\|P_{1})+V_{\textnormal{KL}}(P^*\|P_{2})} \nonumber\\
		& \leq & D_{\text{KL}}(P^{*}\| P_1) + D_{\text{KL}}(P^{*}\| P_2)  
	\end{IEEEeqnarray}
	where the second inequality follows because the third term is nonnegative. Since this contradicts \eqref{eq:ROB_contra1a}, we conclude that any converging subsequence of $\{\tilde{P}_n\}$ must converge to $P^*$, hence $\tilde{P}_n \to P^*$ as $n\to\infty$.

  	\section{Proof of Lemma \ref{Thm:div_ball_converse}}
    \label{sec:append_div_ball_converse}
	Let $\tilde{r}_n$ be as in \eqref{eq:tilde_rn} for an $M_3>0$ to be determined later. For any $(T,R) \in \mathcal{A}_{\bm{\Sigma}_{{\mathbf{P}^*}}}(\tilde{r}_{n})$, we have that 
    \begin{subequations}
	\begin{IEEEeqnarray}{lCl}
	\|T-P^*\|_{2} & = & O\left(\frac{1}{\sqrt{n}}\right)\\
    \|R-P^*\|_{2} & = & O\left(\frac{1}{\sqrt{n}}\right)
	\end{IEEEeqnarray}
    \end{subequations}
	since $\tilde{r}_{n}=\Theta(1/n)$.  Applying the Taylor-series approximation \eqref{eq:div_quad} with $\bm{A}_{D, \mathbf{P}}=\eta \bm{\Sigma}_{\mathbf{P}^*}$,  this   implies that there exist constants $\bar{M}>0$ and $N_3\in\mathbb{N}$ such that, for $n\geq N_3$,
	\begin{equation}
		D(T \| R)  \leq \eta (\mathbf{T}-\mathbf{R})^{T}\bm{\Sigma}_{\mathbf{P^*}}(\mathbf{T}-\mathbf{R}) + \frac{\bar{M}}{n^{3/2}}.  \label{eq:conv0}
	\end{equation}
    We next show that
    \begin{IEEEeqnarray}{lCl}
    \IEEEeqnarraymulticol{3}{l}{(\mathbf{T}-\mathbf{R})^{T}\bm{\Sigma}_{\mathbf{P^*}}(\mathbf{T}-\mathbf{R})} \nonumber\\
    \quad & \leq  & 2 (\mathbf{T}-\mathbf{P}^*)^{T}\bm{\Sigma}_{\mathbf{P^*}}(\mathbf{T}-\mathbf{P}^*) \nonumber\\
    && {}+ 2 (\mathbf{R}-\mathbf{P}^*)^{T}\bm{\Sigma}_{\mathbf{P^*}}(\mathbf{R}-\mathbf{P}^*). \label{eq:conv00}
    \end{IEEEeqnarray}
    Indeed, we have that
	\begin{IEEEeqnarray}{lCl}
		\IEEEeqnarraymulticol{3}{l}{(\mathbf{T}-\mathbf{R})^{T}\bm{\Sigma}_{\mathbf{P^*}}(\mathbf{T}-\mathbf{R})} \nonumber \\
		& = & (\mathbf{T}-\mathbf{P}^{*})^{T}\bm{\Sigma}_{\mathbf{P}^{*}}(\mathbf{T}-\mathbf{P}^*) + (\mathbf{R}-\mathbf{P^*})^{T}\bm{\Sigma}_{\mathbf{P^*}}(\mathbf{R}-\mathbf{P^*})  \nonumber\\
		& & {} -2 (\mathbf{T}-\mathbf{P^*})^{T}\bm{\Sigma}_{\mathbf{P^*}}(\mathbf{R}-\mathbf{P^*}). \label{eq:conv1}
	\end{IEEEeqnarray}
Using the Cauchy-Schwarz inequality, the third term in the right-hand side of \eqref{eq:conv1} can be upper-bounded as
	\begin{IEEEeqnarray}{lCl}
	   \IEEEeqnarraymulticol{3}{l}{-2 (\mathbf{T}-\mathbf{P^*})^{T}\bm{\Sigma}_{\mathbf{P^*}}(\mathbf{R}-\mathbf{P^*})} \nonumber\\
		& \leq & 2  \sqrt{(\mathbf{T}-\mathbf{P^*})^{T}\bm{\Sigma}_{\mathbf{P^*}}(\mathbf{T}-\mathbf{P^*})} \sqrt{(\mathbf{R}-\mathbf{P^*})^{T}\bm{\Sigma}_{\mathbf{P^*}}(\mathbf{R}-\mathbf{P^*})}  \nonumber\\
		& \leq & (\mathbf{T}-\mathbf{P^*})^{T}\bm{\Sigma}_{\mathbf{P^*}}(\mathbf{T}-\mathbf{P^*}) \nonumber\\
        & & {} + (\mathbf{R}-\mathbf{P^*})^{T}\bm{\Sigma}_{\mathbf{P^*}}(\mathbf{R}-\mathbf{P^*}) \label{eq:conv2}
	\end{IEEEeqnarray}
	where the last inequality follows because $2ab \leq a^{2}+b^{2}$ for any real numbers $a$ and $b$. We then obtain \eqref{eq:conv00} from \eqref{eq:conv1} and \eqref{eq:conv2}.
	
	Substituting \eqref{eq:conv00} in \eqref{eq:conv0}, we obtain
	\begin{IEEEeqnarray}{lCl}
		D(T \| R)  &\leq & 2\eta (\mathbf{T}-\mathbf{P}^*)^{T}\bm{\Sigma}_{\mathbf{P^*}}(\mathbf{T}-\mathbf{P^*}) \notag \\ 
         & & {} +2\eta (\mathbf{R}-\mathbf{P}^*)^{T}\bm{\Sigma}_{\mathbf{P}^*}(\mathbf{R}-\mathbf{P}^*)+ \frac{\bar{M}}{n^{3/2}}. \IEEEeqnarraynumspace \label{eq:conv4} 
	\end{IEEEeqnarray}
    Thus, setting $M_3 = \bar{M}/(2\eta)$, it follows that any pair of distributions $(T,R)$ in $\mathcal{A}_{\bm{\Sigma}_{\bar{\mathbf{P}}^*}}(\tilde{r}_{n})$ satisfies
    \begin{equation}
    D(T\|R) \leq r_{n,\epsilon}^{D}.
    \end{equation}
    Hence, $\mathcal{A}_{\bm{\Sigma}_{{\mathbf{P}^*}}}(\tilde{r}_{n}) \subseteq \mathcal{B}_{D}(r_{n,\epsilon}^{D})$.
    
	\section{Proof of Lemma~\ref{Thm:type}}
	\label{sec:appendA_lemmatype}
	Let $P \in \mathcal{P}(\mathcal{Z})$. Consider the pair of probability distributions $(\Gamma', \Gamma'')$ that minimizes $\ell_{P}(T, R)$  (defined in \eqref{eq:minfun1a}) over all $(T,R) \in \mathcal{A}_{\bm{\Sigma}_{\mathbf{P}}}(\tilde{r}_{n})$. That is, for $i=1,\ldots,k$,
	\begin{subequations}
    \begin{IEEEeqnarray}{lCl}
		\Gamma'_{i} & = & P_{i} + x_{i}^{*} \label{eq:minprob1aa}\\
        \Gamma''_{i} & = & P_{i} + y_{i}^{*} \label{eq:minprob1a}
	\end{IEEEeqnarray}
    \end{subequations}
	where
    \begin{subequations}
	\begin{IEEEeqnarray}{lCl}
		x_{i}^{*} & = & \frac{-\sqrt{\tilde{r}_{n}} (\bm{\Sigma}_{\mathbf{P}}^{-1} \mathbf{c})_{i}}{ \sqrt{\mathbf{c}^{\mathsf{T}} \bm{\Sigma}_{\mathbf{P}}^{-1} \mathbf{c} +\mathbf{d}^{\mathsf{T}} \bm{\Sigma}_{\mathbf{P}}^{-1} \mathbf{d}}} \label{eq:minprob2bb}\\
		y_{i}^{*} & = & \frac{-\sqrt{\tilde{r}_{n}} (\bm{\Sigma}_{\mathbf{P}}^{-1} \mathbf{d})_{i}}{ \sqrt{\mathbf{c}^{\mathsf{T}} \bm{\Sigma}_{\mathbf{P}}^{-1} \mathbf{c} +\mathbf{d}^{\mathsf{T}} \bm{\Sigma}_{\mathbf{P}}^{-1} \mathbf{d}}}\label{eq:minprob2b} 
	\end{IEEEeqnarray}
    \end{subequations}
	for $i=1,\ldots,k-1$ and $x_{k}^{*} =-\sum_{i=1}^{k-1}x_{i}^{*}  $, $y_{k}^{*} =-\sum_{i=1}^{k-1}y_{i}^{*}$. The vectors $\mathbf{c}$, $\mathbf{d}$ in \eqref{eq:minprob2bb} and \eqref{eq:minprob2b}  are defined in \eqref{eq:vector_c} and \eqref{eq:vector_d}, respectively,  and we use the notation $(\mathbf{a})_i$ to denote the $i$-th component of a vector $\mathbf{a}$.
	
	To prove the lemma, we need to find a type distribution $(\tp'_{n},\tp''_{n})$ such that, for a sufficiently large $\tilde{N}\in\mathbb{N}$ and all $n\geq\tilde{N}$, the following is true:
	\begin{enumerate}
		\item $(\tp'_{n},\tp''_{n}) \in 	\mathcal{A}_{\bm{\Sigma}_{\mathbf{P}}}(\tilde{r}_{n})$, i.e.,	
		\begin{IEEEeqnarray}{lCl}
			\IEEEeqnarraymulticol{3}{l}{(\mathbf{\tp}'_{n}-\mathbf{P})^{\mathsf{T}} \bm{\Sigma}_{\mathbf{P}} ( \mathbf{\tp}'_{n}-\mathbf{P})  + 	(\mathbf{\tp}''_{n}-\mathbf{P})^{\mathsf{T}} \bm{\Sigma}_{\mathbf{P}} ( \mathbf{\tp}''_{n}-\mathbf{P}) } \nonumber\\
            \qquad & \leq &\tilde{r}_{n}. \label{eq:z9}
		\end{IEEEeqnarray}
		\item For every $ i=1, \ldots,k$,  $n \tp'_{n}(a_{i}) $ and $n\tp''_n(a_i)$ are positive  integers satisfying
		\begin{equation}
			\sum_{i=1}^{k}n \tp'_{n}(a_{i}) =\sum_{i=1}^{k}n \tp''_{n}(a_{i}) =n. \label{eq:z8}  
		\end{equation}
		\item The pair of type distributions $(\tp'_{n},\tp''_{n}) $ satisfies 
		\begin{equation}
			|n \ell_{P}(\Gamma',\Gamma'') -n\ell_{P}(\tp'_{n},\tp''_{n})| \leq \kappa \label{eq:zz}
		\end{equation}
        for some $\kappa > 0$.
	\end{enumerate} 
    
	To prove \eqref{eq:z9}--\eqref{eq:zz}, we write 
	\begin{equation}
		\mathbf{\Gamma'} = \mathbf{P} +\mathbf{x}^{*}, \quad 	\mathbf{\Gamma}'' = \mathbf{P} +\mathbf{y}^{*}
		\label{eq:gamma*}
	\end{equation} 
    and define, for some $0<\bar{\alpha}<1$ to be specified later,
	\begin{equation}
		\mathbf{\bar{\Gamma'}} \triangleq \mathbf{P} +(1-\bar{\alpha})\mathbf{x}^{*}, \quad 	\mathbf{\bar{\Gamma}}'' \triangleq \mathbf{P} +(1-\bar{\alpha})\mathbf{y}^{*}.
		\label{eq:gammabar}
	\end{equation} 
	We then choose  $P'_n$ and $P''_n$ as follows: 
	\begin{subequations}
    \begin{IEEEeqnarray}{lCl}
		n \tp'_{n}(a_{i}) & = &
		\begin{cases}
			\lfloor  n \bar{\Gamma'}_{i} \rfloor, \quad	& \text{if} \;  \langle \bm{\Sigma}_{\mathbf{P}} \mathbf{x}^{*}, \bm{e}_{i} \rangle >0 \\
			\lceil n \bar{\Gamma'}_{i} \rceil, \quad 	&   \text{if}  \;  \langle \bm{\Sigma}_{\mathbf{P}} \mathbf{x}^{*}, \bm{e}_{i} \rangle \leq 0
		\end{cases} 
		\label{eq:tydef1} \\
		n \tp''_{n}(a_{i}) & = &
		\begin{cases}
			\lfloor  n \bar{\Gamma}''_{i} \rfloor, \quad	& \text{if} \;  \langle \bm{\Sigma}_{\mathbf{P}} \mathbf{y}^{*}, \bm{e}_{i} \rangle >0 \\
			\lceil n \bar{\Gamma}''_{i} \rceil, \quad 	&   \text{if}  \;  \langle \bm{\Sigma}_{\mathbf{P}} \mathbf{y}^{*}, \bm{e}_{i} \rangle \leq 0
		\end{cases} \IEEEeqnarraynumspace
		\label{eq:tydef1a}
	\end{IEEEeqnarray}
    \end{subequations}
     where $\lfloor  \cdot \rfloor $ is the floor function; $\lceil \cdot \rceil$ is the ceiling function; $\bm{e}_i=(0,\ldots,0,1,0,\ldots,0)^{\mathsf{T}}$ denotes the standard basis vector in $\mathbb{R}^{k-1}$ whose components are all zero except  at position $i$, where it is one; and $\langle  \cdot, \cdot \rangle $ denotes the dot product in 
$\mathbb{R}^{k-1}$.
    
    In the following, we show that this choice of $(P'_n,P''_n)$ indeed satisfies the conditions~\eqref{eq:z9}--\eqref{eq:zz}. For ease of exposition, we define,
	\begin{equation}
		\bm{\delta'} \triangleq	n \mathbf{\tp}'_{n}-n	\mathbf{\bar{\Gamma'}}, \quad \bm{\delta}'' \triangleq	n \mathbf{\tp}''_{n}-n	\mathbf{\bar{\Gamma}}''. \label{eq:deltadef}
	\end{equation} 
	It follows immediately from~\eqref{eq:tydef1} that 
	\begin{equation}
		|\delta_{i} | <1 \quad \textnormal{and} \quad |\delta_{i} | <1 \label{eq:deltab}
	\end{equation} 
    for $i=1,\ldots,k-1$. Furthermore,
	\begin{subequations}
    \begin{IEEEeqnarray}{lCll}
		\delta'_{i}  & \leq & 0, \quad & \text{if $\langle \bm{\Sigma}_{\mathbf{P}} \mathbf{x}^{*}, \bm{e}_{i} \rangle  >0$}\\
        \delta'_{i} & \geq & 0, \quad  & \text{if $\langle  \bm{\Sigma}_{\mathbf{P}} \mathbf{x}^{*}, \bm{e}_{i} \rangle  \leq 0$} \\
        \delta''_{i} &\leq & 0 , \quad & \text{if $\langle \bm{\Sigma}_{\mathbf{P}} \mathbf{y}^{*}, \bm{e}_{i} \rangle  >0$}   \label{eq:delta1}\\
        \delta''_{i} & \geq & 0 , \quad & \text{if $\langle \bm{\Sigma}_{\mathbf{P}} \mathbf{y}^{*}, \bm{e}_{i} \rangle  \leq 0$}. \label{eq:delta2}
	\end{IEEEeqnarray}
    \end{subequations}
	By \eqref{eq:gammabar} and \eqref{eq:deltadef}, we can express  $(\mathbf{\tp}'_{n},\mathbf{\tp}''_n)$ as
	\begin{equation}
		\mathbf{\tp}'_{n}
		= \mathbf{P} +\mathbf{\bar{x}}, \quad 	\mathbf{\tp}''_{n}
		= \mathbf{P} +\mathbf{\bar{y}} \label{eq:btypdef2} 
	\end{equation} 
	where 
	\begin{equation}
		\mathbf{\bar{x}} 	\triangleq (1-\bar{\alpha})\mathbf{x}^{*} +\frac{\bm{\delta'}}{n}, \quad 	\mathbf{\bar{y}} 	\triangleq (1-\bar{\alpha})\mathbf{y}^{*} +\frac{\bm{\delta}''}{n}. \label{eq:barxdef1} 
	\end{equation} 
	
	\subsection{Proof of  \eqref{eq:z9}}
	Consider
	\begin{IEEEeqnarray}{lCl}
	\IEEEeqnarraymulticol{3}{l}{(\mathbf{\tp}'_{n}-\mathbf{P})^{\mathsf{T}} \bm{\Sigma}_{\mathbf{P}} ( \mathbf{\tp}'_{n}-\mathbf{P})+	(\mathbf{\tp}''_{n}-\mathbf{P})^{\mathsf{T}} \bm{\Sigma}_{\mathbf{P}} ( \mathbf{\tp}''_{n}-\mathbf{P})} \notag \\
		\quad &  = &	\mathbf{\bar{x}}^{\mathsf{T}} \bm{\Sigma}_{\mathbf{P}} \mathbf{\bar{x}}+	\mathbf{\bar{y}}^{\mathsf{T}} \bm{\Sigma}_{\mathbf{P}} \mathbf{\bar{y}}.
	\end{IEEEeqnarray}
	Then, we have
	\begin{IEEEeqnarray}{lCl}
		\IEEEeqnarraymulticol{3}{l}{\mathbf{\bar{x}}^{\mathsf{T}} \bm{\Sigma}_{\mathbf{P}} \mathbf{\bar{x}}+	\mathbf{\bar{y}}^{\mathsf{T}} \bm{\Sigma}_{\mathbf{P}} \mathbf{\bar{y}}} \notag \\
		\quad &  = &(1-\bar{\alpha})^{2} \left( (\mathbf{x}^{*})^{\mathsf{T}} \bm{\Sigma}_{\mathbf{P}} \mathbf{x}^{*}+(\mathbf{y}^{*})^{\mathsf{T}} \bm{\Sigma}_{\mathbf{P}} \mathbf{y}^{*} \right) \notag \\
		& & {} +\frac{1}{n^{2}}	\bm{\delta'}^{\mathsf{T}}\bm{\Sigma}_{\mathbf{P}} \bm{\delta'} +	\frac{1}{n^{2}} \bm{\delta}''^{\mathsf{T}}\bm{\Sigma}_{\mathbf{P}} \bm{\delta}''  \notag \\
		& & {} + \frac{2 (1-\bar{\alpha})}{n} \left(\bm{\delta'}^{\mathsf{T}}\bm{\Sigma}_{\mathbf{P}}\mathbf{x^{*}} +	\bm{\delta}''^{\mathsf{T}}\bm{\Sigma}_{\mathbf{P}}\mathbf{y^{*}} \right). \label{eq:inp}
	\end{IEEEeqnarray}
	After some algebraic manipulations, it can be shown that $(\mathbf{x}^{*})^{\mathsf{T}} \bm{\Sigma}_{\mathbf{P}} \mathbf{x}^{*}+(\mathbf{y}^{*})^{\mathsf{T}} \bm{\Sigma}_{\mathbf{P}} \mathbf{y}^{*}=\tilde{r}_{n}$. 
	Furthermore, using the Rayleigh-Ritz theorem \cite[Th.~4.2.2]{RJ90}, we can upper-bound the second and third terms on the right-hand side of \eqref{eq:inp} as
	\begin{equation}
		\frac{1}{n^{2}}	\bm{\delta'}^{\mathsf{T}}\bm{\Sigma}_{\mathbf{P}} \bm{\delta'} +	\frac{1}{n^{2}} \bm{\delta}''^{\mathsf{T}}\bm{\Sigma}_{\mathbf{P}} \bm{\delta}'' \leq \tilde{\lambda}_{\textnormal{max}} 	 \frac{2(k-1)}{n^{2}} \label{eq:inp2}
	\end{equation}
	where $\tilde{\lambda}_{\textnormal{max}}$ is the maximum eigenvalue of $\bm{\Sigma}_{\mathbf{P}}$. Finally, it can be shown that
	\begin{equation}
		\bm{\delta'}^{\mathsf{T}}\bm{\Sigma}_{\mathbf{P}}\mathbf{x^{*}} +	\bm{\delta}^{\mathsf{T}}\bm{\Sigma}_{\mathbf{P}}\mathbf{y^{*}}
		\leq 0. \label{eq:inp3}
	\end{equation}
	Consequently,
	\begin{IEEEeqnarray}{lCl}
		\IEEEeqnarraymulticol{3}{l}{\mathbf{\bar{x}}^{\mathsf{T}} \bm{\Sigma}_{\mathbf{P}} \mathbf{\bar{x}}+	\mathbf{\bar{y}}^{\mathsf{T}} \bm{\Sigma}_{\mathbf{P}} \mathbf{\bar{y}}} \nonumber\\
		\quad & \leq &(1-\bar{\alpha})^{2} \tilde{r}_{n} +  \frac{\tilde{\lambda}_{\textnormal{max}} 2(k-1)}{n^{2}} \notag \\
		&=& \tilde{r}_{n}+(\bar{\alpha}^{2}\tilde{r}_{n} -\bar{\alpha} \bar{r}_{n}) +\left(  \frac{\tilde{\lambda}_{\textnormal{max}} 2(k-1)}{n^{2}}  -\bar{\alpha} \tilde{r}_{n}\right). \IEEEeqnarraynumspace\label{eq:inp7}
	\end{IEEEeqnarray}
	Since $\tilde{r}_{n}=\Theta(\frac{1}{n})$, there exist $0<\check{M}_{1} \leq \check{M}_{2} < \infty$ and $\check{N}_{1} \in \mathbb{N}$ such that
	\begin{equation}
		\frac{\check{M}_{1}}{ n} \leq \bar{r}_{n} \leq  \frac{\check{M}_{2}}{ n}, \quad n \geq \check{N}_{1}. \label{eq:rno1}
	\end{equation}
	We next choose
	\begin{equation}
		\bar{\alpha}= \frac{2\tilde{\lambda}_{\textnormal{max}} 2(k-1)}{\check{M}_{1} n} \label{eq:rno2}
	\end{equation}
	which is positive and vanishes as $n$ tends to infinity, hence it satisfies $0< \bar{\alpha} <1$ for all $n \geq  \check{N}_{2} $ and a sufficiently large $\check{N}_{2} \geq \check{N}_{1}$. For this choice of $\bar{\alpha}$, we have that
	\begin{equation}
		\left( 	\frac{\tilde{\lambda}_{\textnormal{max}} 2(k-1)}{n^{2}}  -\bar{\alpha} \tilde{r}_{n}\right)  \leq -\frac{\tilde{\lambda}_{\textnormal{max}} 2(k-1)}{n^{2}} \leq 0  \label{eq:rno4}
	\end{equation}
	for $n \geq \check{N}_{2}$. Furthermore, 
	\begin{equation}
		\bar{\alpha}^{2} \tilde{r}_{n}= \Theta \left( \frac{1}{n^{3}}\right) \quad \text{and} \quad  \bar{\alpha}\tilde{r}_{n}= \Theta \left( \frac{1}{n^{2}}\right)
	\end{equation}
	which implies that there exists an $\check{N}_{3} \in \mathbb{N}$ such that
	\begin{equation}
		\bar{\alpha}^{2} \tilde{r}_{n}-\bar{\alpha}\tilde{r}_{n} \leq 0, \quad n \geq \check{N}_{3}. \label{eq:rno5}
	\end{equation}
	Applying \eqref{eq:rno4} and  \eqref{eq:rno5} to \eqref{eq:inp7}, we obtain that, for \mbox{$n \geq \max\left\lbrace  \check{N}_{2},  \check{N}_{3} \right\rbrace $},  
	\begin{equation}
		\mathbf{\bar{x}}^{\mathsf{T}} \bm{\Sigma}_{\mathbf{P}} \mathbf{\bar{x}} +	\mathbf{\bar{y}}^{\mathsf{T}} \bm{\Sigma}_{\mathbf{P}} \mathbf{\bar{y}}
		\leq \tilde{r}_{n}.\label{eq:inp8}
	\end{equation}

    \subsection{Proof of \eqref{eq:z8}}
    The proof of \eqref{eq:z8} follows directly by following the arguments presented in \cite[App.~M]{HJT25}.
    
	\subsection{Proof of \eqref{eq:zz}}
	We have that
    \begin{IEEEeqnarray}{lCl}
		\IEEEeqnarraymulticol{3}{l}{|n \ell_{P}(\Gamma',\Gamma'') -n\ell_{P}(\tp'_{n},\tp''_{n})|} \nonumber\\
        &\leq & \left| \sum_{i=1}^{k-1} nc_{i}x_{i}^{*}-\sum_{i=1}^{k-1}nc_{i} \bar{x}_{i} \right|  +\left| \sum_{i=1}^{k-1} nd_{i}y_{i}^{*}-\sum_{i=1}^{k-1}nd_{i} \bar{y}_{i} \right|. \IEEEeqnarraynumspace
	\end{IEEEeqnarray}
	By taking $\kappa  \triangleq (\widehat{M}+1)(\sum_{i=1}^{k-1}|c_{i}|+\sum_{i=1}^{k-1}|d_{i}|)$, the remaining steps of the proof follow directly from the arguments presented in \cite[App.~M]{HJT25}.
	
	\section{Proof of Proposition~\ref{prop:gendiv}}
	\label{Append:bla}
	It was shown in \cite[Th.~10]{Shengyu21} that no two-sample test with type-I error bounded by $\epsilon$ can achieve a type-II first-order term $\beta'$ that exceeds $2 D_B(P_1,P_2)$. More precisely, the proof of \cite[Th.~10]{Shengyu21} implies that
\begin{IEEEeqnarray}{lCl}
		\varlimsup_{n\to\infty} -\frac{\ln	\beta_n(\mathsf{T}_n)}{n} & \leq &   2 D_{B}(P_{1},P_{2}) \label{eq:(140)}
	\end{IEEEeqnarray}
    for any two-sample test $\mathsf{T}_n$ satisfying $\varlimsup_{n\to\infty} \alpha_n(\mathsf{T}_n) \leq \epsilon$. Proposition~\ref{prop:gendiv} follows then by showing that there is a divergence test $\textsf{T}_n^{D}(r_n)$ and a sequence of thresholds $\{r_n\}$ that achieve
    \begin{equation}
    \label{eq:(140)_better}
		\lim_{n\to\infty} \alpha_n\left(\mc_n^{D}(r_n)\right) = 0
	\end{equation}
    and
    \begin{equation}
\varliminf_{n\to\infty} -\frac{\ln	\beta_n(\mathsf{T}_n)}{n} \geq 2 D_{B}(P_{1},P_{2}). \label{eq:(140)_even_better}
    \end{equation}
    To this end, we choose some thresholds $\{r_n\}$ satisfying
	\begin{equation}
    \label{eq:prop:gendiv_proof_rn}
		r_n = \omega\left(\frac{1}{n}\right) 
		\quad \text{and} \quad 
		r_n = o(1).
	\end{equation}
	By Lemma~\ref{cov_lemma}, we have that
	\begin{equation}
		P^{n}\left( \frac{n}{2} D\left(\tp_{X^{n}} \middle\| \tp_{Y^{n}}\right) \geq c \right)
		= \mathsf{Q}_{\chi^{2}_{\bm{\lambda},k-1}}(c) + O(\delta_n)
	\end{equation}
    for all $c>0$ and some positive sequence $\{\delta_n\}$ that is independent of $c$ and vanishes as $n\to\infty$.
	Consequently, the type-I error satisfies
	\begin{IEEEeqnarray}{lCl}
		\alpha_n\left(\mc_n^{D}(r_n)\right)
		&=& P^{n}\left( \frac{n}{2} D\left(\tp_{X^{n}} \middle\| \tp_{Y^{n}}\right) \geq \frac{n r_n}{2} \right) \notag \\
		&=& \mathsf{Q}_{\chi^{2}_{\bm{\lambda},k-1}}\left( \frac{n r_n}{2} \right) + O(\delta_n).
	\end{IEEEeqnarray}
	Since $r_n = \omega(1/n)$, we have that $n r_n \to \infty$ as $n\to\infty$, so \eqref{eq:(140)_better} follows.

 	
	We next upper-bound the type-II error $\beta_n(\textsf{T}_n^{D}(r_n))$ using the method of types  \cite[Th.~11.1.4]{ATCB} to obtain
	\begin{IEEEeqnarray}{lCl}
		\IEEEeqnarraymulticol{3}{l}{\beta_n(\textsf{T}_n^{D}(r_n)) } \nonumber\\
		&\leq & \sum_{(P',P'')\in\mathcal{B}_{D}(r_n)\cap(\mathcal{P}_n\times\mathcal{P}_n)}
		e^{-n D_{\mathrm{KL}}(P'\|P_1) -n D_{\mathrm{KL}}(P''\|P_2)}. \nonumber\\ \label{eq:(143)}
	\end{IEEEeqnarray}
We then lower-bound $D_{\mathrm{KL}}(P'\|P_1)+D_{\mathrm{KL}}(P''\|P_2)$ over $(P',P'')\in\mathcal{B}_{D}(r_n)$ by proceeding as in Section~\ref{sub:UB}. Indeed, since $r_n=o(1)$, it follows from \cite[Lemma~2]{HJT25} that, for every pair $(P',P'')$ in $\mathcal{B}_{D}(r_n)$, we have $\|\mathbf{P}'-\mathbf{P}''\|_{2} = o(1)$ since  $\|\mathbf{P}'-\mathbf{P}''\|_{2} = O(\sqrt{r_n})$. This further implies that $\|P'-P''\|_{1}=o(1)$. We then set $\bar{P}=\frac{1}{2}P' + \frac{1}{2} P''$ and distinguish between the cases $(\bar{P}, \bar{P}) \notin\mathcal{B}_{\ell_1,P^*}(\delta)$ and $(\bar{P}, \bar{P}) \in \mathcal{B}_{\ell_1,P^*}(\delta)$, where $\mathcal{B}_{\ell_1,P^*}(\delta)$ is as in \eqref{eq:BL1} with a $\delta>0$ chosen so that $\mathcal{B}_{\ell_1,P^*}(\delta)$ is contained in $\mathcal{P}^2(\mathcal{Z})$.

If $(\bar{P}, \bar{P}) \notin\mathcal{B}_{\ell_1,P^*}(\delta)$, then we obtain directly from \eqref{eq:(66)} that
\begin{IEEEeqnarray}{lCl}
    \IEEEeqnarraymulticol{3}{l}{\inf_{(P', P'')\in \mathcal{B}_{D}(r_n)} \{D_{\text{KL}}(P'\| P_{1})+D_{\text{KL}}(P'' \| P_{2})\}} \notag \\
    \quad & >  & D_{\text{KL}}(P^{*}\| P_{1})+D_{\text{KL}}(P^{*} \| P_{2}).\label{eq:outside}
\end{IEEEeqnarray}
If $(\bar{P}, \bar{P}) \in \mathcal{B}_{\ell_1,P^*}(\delta)$, then we perform a Taylor-series approximation of $D_{\text{KL}}(P'\|P_1) + D_{\text{KL}}(P''\|P_2)$ around $(\bar{P},\bar{P})$ to obtain
\begin{IEEEeqnarray}{lCl}
\IEEEeqnarraymulticol{3}{l}{D_{\text{KL}}(P'\|P_1) + D_{\text{KL}}(P''\|P_2)} \nonumber\\
\quad & \geq & D_{\text{KL}}(\bar{P}\|P_1) + D_{\text{KL}}(\bar{P}\|P_2) + o(1)
\end{IEEEeqnarray}
where the $o(1)$-term is independent of $\bar{P}$. It then follows that
\begin{IEEEeqnarray}{lCl}
    \IEEEeqnarraymulticol{3}{l}{\inf_{(P', P'')\in \mathcal{B}_{D}(r_n)} \{D_{\text{KL}}(P'\| P_{1})+D_{\text{KL}}(P'' \| P_{2})\}} \notag \\
    \quad & \geq  & \min_{\bar{P}\in\mathcal{P}(\mathcal{Z})} \bigl\{ D_{\text{KL}}(\bar{P}\|P_1) + D_{\text{KL}}(\bar{P}\|P_2) \bigr\} + o(1) \nonumber\\
    & = & D_{\text{KL}}(P^{*}\| P_{1})+D_{\text{KL}}(P^{*} \| P_{2}) + o(1) \label{eq:inside}
\end{IEEEeqnarray}
where the last step follows because $P^*$ is the minimizer of $D_{\text{KL}}(\bar{P}\|P_1)+D_{\text{KL}}(\bar{P}\|P_2)$.

Combining \eqref{eq:outside} and \eqref{eq:inside}, and recalling that $D_{\text{KL}}(P^{*}\| P_{1})+D_{\text{KL}}(P^{*} \| P_{2})=2D_B(P_1,P_2)$, we obtain that
\begin{IEEEeqnarray}{lCl}
    \IEEEeqnarraymulticol{3}{l}{\inf_{(P', P'')\in \mathcal{B}_{D}(r_n)} \{D_{\text{KL}}(P'\| P_{1})+D_{\text{KL}}(P'' \| P_{2})\}} \nonumber\\
    \quad & \geq & 2 D_B(P_1,P_2) + o(1). \label{eq:inside_outside}
\end{IEEEeqnarray}
Applying \eqref{eq:inside_outside} to \eqref{eq:(143)}, and using that the number of types $|\mathcal{P}_n|$ is bounded by $(n+1)^{|\mathcal{Z}|}$ \cite[Th.~11.1.1]{ATCB},  we can then lower-bound
	\begin{IEEEeqnarray}{lCl}
		-\frac{\ln \beta_n(\textsf{T}_n^{D}(r_n))}{n} & \geq & 2 D_B(P_1,P_2) + o(1).
	\end{IEEEeqnarray}
    This yields \eqref{eq:(140)} and proves the proposition.

\balance
\bibliography{Bibliography.bib}

@book{AB216,
	title={Information Geometry and Its Applications},
	author={Amari, Shun-Ichi},
	volume={194},
	year={2016},
	publisher={Springer}
}

@article{H65,
	title={Asymptotically optimal tests for multinomial distributions},
	author={Hoeffding, Wassily},
	journal={Ann. Math. Stat.},
	pages={369--401},
	month={Apr.},
	year={1965}
}

@book{CZPA209,
	title={Nonnegative Matrix and Tensor Factorizations: Applications to Exploratory Multi-Way Data Analysis and Blind Source Separation},
	author={Cichocki, Andrzej and Zdunek, Rafal and Phan, Anh Huy and Amari, Shun-Ichi},
	year={2009},
	publisher={John Wiley \& Sons}
}

@book{ATCB,
	author = {Cover, Thomas and Thomas, Joy},
	title = {Elements of Information Theory},
	year = {2006},
	edition={2nd ed.},
	publisher = {Wiley-Interscience},
	address={New York, NY, USA},
}

@book{BDB91,
	title={Time Series: Theory and Methods},
	author={Peter J.. Brockwell and Davis, Richard A},
	year={1991},
	publisher={Springer-Verlag}
}

@book{CS204,
	title={Information Theory and Statistics: A Tutorial},
	author={Csisz{\'a}r, Imre and Shields, Paul C},
	year={2004},
	publisher={Now Publishers Inc}
}

@ARTICLE{Unnikrishnan16,
	author={Unnikrishnan, Jayakrishnan and Huang, Dayu},
	journal={IEEE Trans. Inf. Theory}, 
	title={Weak Convergence Analysis of Asymptotically Optimal Hypothesis Tests}, 
	year={2016},
	volume={62},
	number={7},
	month={Jul.},
	pages={4285-4299},
}

@Book{RC199,
	Title                    = {Probability and Measure Theory},
	Author                   = {Robert B. Ash and Catherine A. Doleans-Dade},
	Publisher                = {Harcourt/Academic Press},
	Year                     = {1999},
	Edition                  = {Second}
}

@book{RJ90,
	title={Matrix Analysis},
	author={Horn, Roger A and Johnson, Charles R},
	year={1990},
	publisher={Cambridge University Press}
}

@article{CL86,
	title={An extended {{\v{C}}encov} characterization of the information metric},
	author={Campbell, L Lorne},
	journal={Proc. Am. Math. Soc.},
	volume={98},
	number={1},
month = {Sep.},
	pages={135--141},
	year={1986}
}

@ARTICLE{Gut89,
  author={Gutman, M.},
  journal={IEEE Trans. Inf. Theory}, 
  title={Asymptotically optimal classification for multiple tests with empirically observed statistics}, 
  year={1989},
  volume={35},
  number={2},
month = {Mar.},
  pages={401-408},}

@ARTICLE{Shengyu21,
  author={Zhu, Shengyu and Chen, Biao and Chen, Zhitang and Yang, Pengfei},
  journal={IEEE Trans. Inf. Theory}, 
  title={Asymptotically Optimal One- and Two-Sample Testing With Kernels}, 
  year={2021},
  volume={67},
  number={4},
    month = {Apr.},
  pages={2074-2092},}

@article{Lin19,
    author = {Zhou, Lin and Tan, Vincent Y F and Motani, Mehul},
    title = {Second-order asymptotically optimal statistical classification},
    journal = {Information and Inference: A Journal of the IMA},
    volume = {9},
    number = {1},
    pages = {81-111},
    year = {2020},
    month = {Mar.},
}

@ARTICLE{HJT25,
  author={Harsha, K. V. and Ravi, Jithin and Koch, Tobias},
  journal={IEEE Trans. Inf. Theory}, 
  title={On the Second-Order Asymptotics of the {Hoeffding} Test and Other Divergence Tests}, 
  year={2025},
  volume={71},
  number={10},
month = {Oct.},
  pages={7459-7483},}

@book{LR205,
  title={Testing Statistical Hypotheses},
  author={Lehmann, Erich Leo and Romano, Joseph P},
  year={2005},
  publisher={Springer}
}

@ARTICLE{Li_Universal,
author={Y. Li and S. Nitinawarat and V. V. Veeravalli},
journal={IEEE Trans. Inf. Theory},
title={Universal Outlier Hypothesis Testing},
year={2014},
volume={60},
number={7},
pages={4066-4082},
month={Jul.},}

@article{Gretton12,
  author  = {Arthur Gretton and Karsten M. Borgwardt and Malte J. Rasch and Bernhard Sch{{\"o}}lkopf and Alexander Smola},
  title   = {A Kernel Two-Sample Test},
  journal = {J. Mach. Learn. Res.},
  year    = {2012},
  volume  = {13},
  number  = {25},
  pages   = {723-773}
}

@book{PW25,
  title={Information Theory: From Coding to Learning},
  author={Polyanskiy, Yury and Wu, Yihong},
  year={2025},
  publisher={Cambridge University Press}
}
\bibliographystyle{IEEEtran}

\end{document}